\documentclass[a4paper,12pt]{article}
\usepackage{amsmath,amssymb,amsthm}
\usepackage[colorlinks,citecolor=blue,urlcolor=blue,breaklinks=true]{hyperref}
\usepackage{amsfonts}
\usepackage{makeidx}
\usepackage{multirow}
\usepackage{lscape}
\usepackage{float}
\usepackage{graphicx}
\usepackage{lineno}
\usepackage{enumerate}
\usepackage[authoryear,round]{natbib}

\makeatletter

\newcommand{\Rmnum}[1]{\expandafter\@slowromancap\romannumeral #1@}
\makeatother

\newcommand{\E}{\mathop{\mathbb{E}}}
\newcommand{\R}{\mathop{\mathbb{R}}}
\newcommand{\la}{\lambda}
\newcommand{\pr}{\textrm{pr}}

\def \sumi {\displaystyle\sum_{i=0}^{\infty}}

\def \E{\mbox{E}}

\makeatletter
\newcommand*\nobreakhyphen{\hbox{-}\nobreak\hskip\z@skip}
\makeatother

\newtheorem{theorem}{Theorem}[section]
\newtheorem{lemma}[theorem]{Lemma}

\newtheorem{corollary}[theorem]{Corollary}
\newtheorem{definition}[theorem]{Definition}
\newenvironment{remark}[1][Remark]{\begin{trivlist}
\item[\hskip \labelsep {\bfseries #1}]}{\end{trivlist}}

\title{Self-excited Threshold Poisson Autoregression}
\author{
Chao Wang,
\and
Heng Liu, 
\and
Jian-Feng Yao,
\and
Richard A. Davis,
\and
Wai Keung Li 
\footnote{
This is a second revision; the original paper was written in November 2011 and the first revision in June 2012.
Chao Wang is Post-doctoral Fellow, Department of Statistics and Actuarial Science, the University of Hong Kong, Hong Kong (email: chaowang@connect.hku.hk);
Heng Liu is PhD Student, Department of Statistics, Columbia University, New York (email: hengliu@stat.columbia.edu);
Jian-Feng Yao is Associate Professor, Department of Statistics and Actuarial Science, the University of Hong Kong, Hong Kong (email: jeffyao@hku.hk);
Richard A. Davis is Professor, Department of Statistics, Columbia University, New York (email: rdavis@stat.columbia.edu);
and 
Wai Keung Li is Professor, Department of Statistics and Actuarial Science, the University of Hong Kong, Hong Kong (email: hrntlwk@hku.hk).
This research is conducted using the HKU Computer Centre research computing facilities supported partially by Hong Kong UGC Special Equipment Grant SEG HKU09.
Richard A. Davis's research is supported in part by the U.S. National Science Foundation grant DMS-1107031.
Wai Keung Li's research is supported partially by the HKSAR Research Grant Council General Research Fund \#703711.
}
}

\begin{document}
\date{June 29, 2013}
\maketitle
\begin{abstract}
This paper studies theory and inference of an observation-driven model for time series of counts. It is assumed that the observations follow a Poisson distribution conditioned on an accompanying intensity process, which is equipped with a two-regime structure according to the magnitude of the lagged observations. The model remedies one of the drawbacks of the Poisson autoregression model by allowing possibly negative correlation in the observations. Classical Markov chain theory and Lyapunov's method are utilized to derive the conditions under which the process has a unique invariant probability measure and to show a strong law of large numbers of the intensity process. Moreover the asymptotic theory of the maximum likelihood estimates of the parameters is established. A simulation study and a real data application are considered, where the model is applied to the number of major earthquakes in the world.
\end{abstract}


\section{Introduction}
There has been increasing interest in developing models for time series of counts because of their wide range of applications, including epidemiology, finance, disease modeling and environmental science. The majority of these models assume that the observations follow a Poisson distribution conditioned on an accompanying intensity process that drives the dynamics of the model, see \citet{DavisDunsmuirStreett2003}, \citet{FerlandLatourOraichi2006}, \citet{FokianosRahbekTjostheim2009}, \citet{FokianosTjostheim2011}, \citet{DavisLiu2012} and \citet{DoukhanFokianosTjostheim2012}. According to whether the evolution of the intensity process depends on the observations or solely on an external process, \cite{Cox1981} classified the models into observation-driven and parameter-driven. Compared to parameter-driven models, an observation-driven model usually enjoys a considerably easier and more straightforward estimation procedure, however, it is difficult to establish stability properties, including stationarity and mixing conditions of the model. This paper formulates and investigates a self-excited threshold Poisson autoregression process, which belongs to the class of observation-driven models.

One observation-driven model, the Poisson autoregression, also known as the Poisson integer-valued GARCH (INGARCH), has already received considerable study in the literature, see for example, \cite{FerlandLatourOraichi2006}, \cite{FokianosRahbekTjostheim2009}, \cite{Neumann2011}, \cite{DoukhanFokianosTjostheim2012}, \cite{DavisLiu2012}, and \citet{FokianosTjostheim2012}. For this model, it is assumed that the observations $\{Y_t\}$ given the intensity process $\{\lambda_t\}$ follow Poisson distribution, where $\lambda_t$ follows the GARCH-like recursions $\lambda_t=\delta + \alpha \lambda_{t-1}+\beta Y_{t-1}$. The name GARCH associated with this model comes from \citet{Bollerslev1986} as the Poisson mean coincides with its variance, and is known for its capability of capturing positive temporal dependence in the observations and it is relatively easy to fit via maximum likelihood. \cite{FokianosRahbekTjostheim2009} studied the model and established the asymptotic theory of the parameter estimates by introducing a small perturbation. \cite{Neumann2011} considered some contracting dynamics of $\lambda_t$ and derived mixing condition of the count process. \cite{DavisLiu2012} generalized the conditional distribution of $\{Y_t\}$ to a one-parameter exponential family and took advantage of the theory for iterated random functions \citep{DiaconisFreedman1999,WuShao2004} to establish stationarity and absolute regularity of the process, as well as the asymptotic distribution of the parameter estimates. \cite{DoukhanFokianosTjostheim2012} showed similar results by utilizing the concept of $\tau$-weak dependence. More recently, \cite{BlasquesKoopmanLucas2012} considered a class of generalized autoregressive score processes which includes Poisson autoregression as a special case and used the Dudley entropy integral to obtain a wider non-degenerate parameter region that guarantees the stationarity and ergodicity of the processes.

Despite many advantages that the Poisson autoregression model enjoys, it is incapable of modeling negative serial dependence in the observations. This can be seen through the fact that $\{Y_t\}$ can be represented as an ARMA$(1, 1)$ process with a sequence of martingale differences as innovations and with a positive autoregressive coefficient (see e.g., \cite{DavisLiu2012}). This concern motivated \citet{FokianosTjostheim2011} in part to study the so-called log-linear Poisson autoregression. Our paper proposes a self-excited threshold integer-valued Poisson autoregression model (SETPAR), which allows for a more general modeling framework for the intensity process, including the possibility of negative serial dependence in the data. The model assumes a two-regime structure of the conditional mean process $\{\lambda_t\}$ according to the magnitude of the lagged observations. Such an extension to a model with threshold has its own merits, on account of the successful modeling strategy of a self-excited threshold autoregressive moving average process introduced by \cite{Tong1990}. 

Some studies have been directed to this model from different perspectives. \citet{WoodardMattesonHenderson2011} discussed a large class of the so-called ``generalized autoregressive moving average models'' which includes a similar threshold model. The model was also found in another general study of observation-driven time series models by \citet{DoucDoukhanMoulines2013}. Despite several similar results found in their papers and ours, we adopt a different methodology, which is well suited to these types of models.  The difficulty with the theory is that the Markov kernel associated with the model lacks proper continuity. \citet{WoodardMattesonHenderson2011} adopted the existing approach of \citet{FokianosRahbekTjostheim2009} which is based on a smoothed approximation of the Markov chain by adding an asymptotically vanishing noise. \citet{DoucDoukhanMoulines2013} considered the model directly and applied a coupling construction to prove the uniqueness of the stationary distribution with the same conditions on model coefficients for the ergodicity as ours (compare their Proposition 14 and our Theorem~\ref{thm:stable}). We studied the model directly using a different concept of e-chain (see Chapter 6, \citet{MeynTweedie1993}), which has an asymptotic continuity property that guarantees the uniqueness of a stationary distribution with mild additional conditions. Regarding the coverage of the approaches, the coupling argument applies to the log-linear Poisson autoregressions \citep{FokianosTjostheim2011, DoucDoukhanMoulines2013} as well. This is however not surprising since the Markov chains in a log-linear Poisson autoregressions and SETPAR model are very similar and our approach through e-chains can also be used for a log-linear Poisson autoregression as well. In addition, we are able to establish consistency and asymptotic normality of the maximum likelihood estimates directly based on our discussion of the stability property of the model under mild conditions on the parameters.

The organization of the paper is as follows. Section~\ref{sec:model} formulates the model and establishes its stability properties. Likelihood inference and asymptotic theory of the estimates are investigated in Section~\ref{estimation}. Some numerical results, including a simulation study and a real data example are given in Section~\ref{sec:sim}. The model is applied to the counts of major earthquakes in the world, and some diagnostic tools for assessing and comparing model performance are also given in this section. Section~\ref{sec:concl} discusses some problems which are worth further study and concludes the paper. Proofs of the key results in Sections~\ref{sec:model} and Section~\ref{estimation} are deferred to the Appendix.

\section{The model and its properties}
\label{sec:model}

For ease of discussion, only the first order self-excited threshold Poisson autoregression is investigated in this paper. However, the generalization to higher order model with multiple thresholds is also possible using similarly stylized arguments.

\begin{definition}
  \label{def:tpar}  
A sequence of random observations $\left\{ Y_t,t\in\mathbb{Z} \right\}$ is said to follow the self-excited threshold Poisson autoregression (SETPAR) model, if
\begin{align}
\mathcal{L}(Y_t \mid \mathcal{F}_{t-1}) & = \textrm{Poisson}(\lambda_t),
\label{y1}
\end{align}
where $\mathcal{F}_{t}=\sigma\left\{ Y_s,~s\leq t \right\}$, and
\begin{align}
  \lambda_t=
\left\{
\begin{array}{cc}
  d_1+a_{1}\lambda_{t-1}+b_{1}Y_{t-1},&Y_{t-1}\leq  r,\\
\\
  d_2+a_{2}\lambda_{t-1}+b_{2}Y_{t-1},&Y_{t-1} > r,\\
\end{array}
\right.
\label{lambda}
\end{align}
with $d_{i}> 0, a_{i}> 0,b_{i}> 0,~i=1,2$, and $r\in \mathbb{N}$. 
\end{definition}


Let $\theta^{(i)}=(d_{i},a_{i},b_{i})^\intercal\; (i=1,2)$ be the regime-specific parameter vector. It is reasonable to assume $\theta^{(1)}\ne \theta^{(2)}$, since otherwise, the model is reduced to the ordinary Poisson autoregression. The intercept parameter $d_i$ is restricted to be positive to avoid a Poisson distribution with zero mean. 

The dynamics of the process is governed by a two-regime scheme. In the following context, if $Y_{t-1}\leq r$ then we say $Y_t$ lies in the lower regime, denoted by $Y_t\in R_1$, where $R_1=\left\{ 0,\dots,r \right\}$; otherwise, $Y_t$ is in the upper regime, denoted by $Y_t\in R_2$, $R_2=\mathbb{N}-R_1$.

Let $\left\{ N_t(\cdot),t\in\mathbb{Z} \right\}$ be a sequence of independent Poisson processes with unit intensity. As suggested by \citet{FokianosRahbekTjostheim2009}, it is sometimes convenient to treat $Y_t$ in Eq~\eqref{y1} as the sampling value of $N_t$ at time $\lambda_t$, i.e.,
\begin{align}
  Y_t=N_t(\lambda_t),
  \label{y}
\end{align}
where $\lambda_t$ is the same as in Eq~\eqref{lambda}.

Although the process $\left\{ \lambda_t \right\}$ as well as the joint one $\left\{ (\lambda_t,Y_t) \right\}$  is a Markov chain, it is difficult to investigate the properties of the these processes, mainly due to the fact that the real-valued intensity process $\lambda_t$ is a function of the real-valued $\lambda_{t-1} $ and the discrete-valued innovations $Y_{t-1}$ (see also \citet{FokianosRahbekTjostheim2009}, \citet{WoodardMattesonHenderson2011}). 
In particular, it is easy to show that $\{\lambda_t\}$ is not a strong Feller chain even for the Poisson autoregression model without a threshold, which implies that one needs to apply more nonstandard Markov chain theory, such as Lyapunov's method and e-chains, in order to establish stability properties. Due to the importance of the concept of stability, its definition by \citet{Duflo1997} is given below. 
Readers are referred to Sections 6.1-6.2 in \citet{Duflo1997} and Section 6.4 in \citet{MeynTweedie1993} for other corresponding definitions and relevant theory of Lyapunov's method and e-chains.

\begin{definition}
  (Definition 6.1.1, Definition 6.1.4, \citet{Duflo1997}) Suppose that a random sequence $\left\{ X_n \right\}$ is defined on a metric space $E$ together with its Borel $\sigma$-field. $\left\{ X_n \right\}$ is said to be a stable model if there exists a probability distribution $\mu$ on $E$ such that, for almost all $\omega$, the sequence of empirical distributions $$\Lambda_n(\omega,\cdot)=\frac{1}{n+1}\sum_{t=0}^n 1\left\{ X_t(\omega)\in \cdot \right\}$$ converges weakly to $\mu$. The distribution $\mu$ is the stationary distribution for the model.
   
   A Markov chain is said to be stable if its state space is a metric space, and for any initial distribution $\nu$, the induced random sequence is stable with a stationary distribution independent of $\nu$.
\end{definition}



We begin with the following theorem establishing the stability of $\{\lambda_t\}$.

\begin{theorem}
  Consider the model in Definition~\ref{def:tpar}. Assume $a_1<1$ and $a_2+b_{2}<1$ , then
  \begin{enumerate}
  \item 
    The Markov chain $\left\{ \lambda_t \right\}$ is stable and possesses a unique invariant probability measure $\mu$, which has moments of all orders.
  \item 
    For any $\mu$-a.s. continuous function  $\phi$ satisfying 
    \[    |\phi(\lambda)| \le c(1+\lambda^k),
    \]
    for some power $k\ge0$ and constant $c$, 
    it holds that
    \[     \frac1n \left[ \phi(\lambda_1)+\cdots+\phi(\lambda_n)  \right]\to \mu(\phi), a.s.
    \]
   for any initial value $\lambda_0$.
  \end{enumerate}
\label{thm:stable}
\end{theorem}

The properties of the observed process $\left\{ Y_t \right\}$ can be deduced from the properties of $\left\{ \lambda_t \right\}$, as stated in the following corollary.

\begin{corollary}
  Suppose the assumptions of Theorem~\ref{thm:stable} hold, then the joint process $\left\{ (\lambda_t,Y_t) \right\}$ is stable and $\left\{ Y_t \right\}$ has finite moments of all orders.
  \label{prop:stable_y}
\end{corollary}

Similar to Theorem~\ref{thm:stable}, the stability of the joint process ensures the law of large numbers holds for polynomial functions of $(\lambda_t,Y_t)$, which serves an important role in establishing the asymptotic theory of the estimators for the parameters in next section.

As is claimed that this model can produce negative autocorrelation, we conclude this section by some remarks on the autocorrelation function of this model. It turns out that an explicit formula of its autocorrelation function is very difficult to obtain, and to our best knowledge, no such result exists for time series models with thresholds. Based on the stability of the model, the claim can be verified by Monte Carlo simulations, since the sample autocorrelation is a consistent estimator for the theoretical autocorrelation. 
As to the theoretical property of the autocorrelation function, it can be proved that when $b_1$ is large enough, $\E\left( \lambda_{t}| \lambda_{t-1} \right)$ is a decreasing function of $\lambda_{t-1}$. Thus, it is likely that  $\lambda_t$ and $\lambda_{t-1}$ will vary in opposite directions with high probability and the pair $(Y_{t},Y_{t-1}) $ will display a negative correlation as $Y_t=N_t(\lambda_t)$ and $Y_{t-1}=N_{t-1}(\lambda_{t-1})$.

\section{Parameter estimation by maximum likelihood}
\label{estimation}

Suppose we have a series of observations $\left\{ Y_t \right\}_{t=1}^n$ generated from the self-excited threshold Poisson autoregression model and we want to estimate the parameters. Feasible approaches include the least squares estimator and the maximum likelihood estimator. Since the likelihood function for given observations $\left\{ Y_t \right\}_{t=1}^n$ can be easily calculated with an initial value of $\lambda_1$ and the maximum likelihood estimator is likely to be more efficient than the least square estimator, we only discuss the maximum likelihood estimator here.

Recall that $\theta^{(i)}=(d_{i},a_{i},b_{i})^\intercal$ is the parameter vector for the $i^{th}$ regime, $i=1,2$. Then $\theta=(r,\theta^{(1)}{}^\intercal ,\theta^{(2)}{}^\intercal)^\intercal$ denotes the vector of all parameters. Let $\theta_0$ be the true parameter vector.
Let $\lambda_{t,i}=d_{i}+a_{i}\lambda_{t-1}+b_{i}Y_{t-1}\; (i=1,2)$,
then $\lambda_{t}=\sum_i \lambda_{t,i}1\left\{ Y_t\in R_i \right\}$. Since the $\lambda_t$'s have to be calculated recursively, an initial value $\lambda_1$ is needed.

Fix an arbitrary initial value of $\lambda_1$, denoted by $\tilde{\lambda}_1$. Let $\{ \tilde{\lambda}_t \}_{t=2}^n$ be the sequence calculated by the recursive equation Eq~\eqref{lambda} with the initial value $\tilde{\lambda}_1$ and the observed data $\left\{ Y_t \right\}_{t=1}^n$. Then the log-likelihood function, apart from a constant, is
\begin{align*}
  \tilde{\ell}(\theta)&=\sum_{t=1}^n \tilde{\ell}_t(\theta),
\end{align*}
where
  $\tilde{\ell}_t=-\tilde{\lambda}_{t}+Y_t \log(\tilde{\lambda}_{t})$.

The maximum likelihood estimator of $\theta$ is
\begin{align}
  \hat{\theta}=\arg \max_{\theta\in ([0,r_*]\cap \mathbb{N})\times\mathcal{D} } \tilde{\ell} (\theta),\label{theta_hat}
\end{align}
where $r_*$ is some large positive integer and $\mathcal{D}$ is some compact subset of $\mathbb{R}^6$ which will be specified later.

To study the asymptotic behaviour of the estimator, we make the following assumption about the underlying process and the parameter space.\\
{\sc Assumption:}
\begin{itemize}
  \item [(A1)] The observed sequence $\left\{ Y_t \right\}_{t=1}^n$ is generated from the self-excited threshold Poisson autoregression process, with true parameter $\theta_0 \in ([0,r_*]\cap \mathbb{N})\times \mathcal{D}^o $, where $\mathcal{D}^o$ is the interior of $\mathcal{D}\subset \mathbf{\Theta}$, and $\mathbf{\Theta}=\{(d_1,a_1,b_1,d_2,a_2,b_2)^\intercal\in \R_+^6:~a_1<1,~b_1<1,~a_2+b_2<1\}$, where $\R_+$ is the strictly positive part of the real line.
\end{itemize}

\begin{remark}
The assumptions are quite natural and broad. Note the restriction of the parameters in the lower regime. Although it is shown in Corollary~\ref{prop:stable_y} that the joint process $\left\{ (\lambda_t,Y_t) \right\} $ is stable for any $b_1>0$, currently it is necessary to assume $b_1<1$ when proving the asymptotic properties of the maximum likelihood estimators. We conjecture that the same asymptotic properties would hold for parameters with $b_1\geq 1$ under other assumptions but leave it for future study. 
Nevertheless, the restricted parameter space still contains some explosive lower regime in the sense that $a_1+b_1>1$.
\end{remark}

Bearing in mind that the calculation of the log-likelihood $\tilde{\ell}(\theta)$ is based on an initial value of $\lambda_1$, in order to establish the asymptotic properties of $\hat{\theta}$, we need to show that the effect of selecting different initial value $\tilde{\lambda}_1$ is asymptotically negligible. 

To see this, note that the process can also be represented as a varying-coefficient Poisson autoregression model in the sense that the coefficients of the Poisson autoregression model vary with the past observation. Specifically, for a given parameter vector $\theta$, let $d_{t}=\sum_{i=1}^2 d_i 1\left\{ Y_t\in R_i \right\}$, $a_{t}=\sum_{i=1}^2 a_{i} 1\left\{ Y_t\in R_i \right\}$ and $b_{t}=\sum_{i=1}^2 b_{i} 1\left\{ Y_t\in R_i \right\}\;(t=1,\dots, n)$, assuming that no ambiguity shall be caused by the notation of $a_t$ and $b_t$ for $t=1,2$. Then $\lambda_t=\lambda_t(\theta)$ satisfies the recursive equation,
\begin{align}
    \lambda_t&=d_{t-1}+b_{t-1}Y_{t-1}+a_{t-1}\lambda_{t-1}\label{lambda_vc}\\
    &:=c_{t-1}+a_{t-1}\lambda_{t-1}\label{lambda_re}\\
    &=\sum_{k=1}^\infty \prod_{j=1}^{k-1}a_{t-j} c_{t-k} \label{lambda_sol}.
\end{align}

Eq~\eqref{lambda_re} defines a recursive equation of $\lambda_t$ assuming the process $\{Y_t\}$ and the vector $\theta$ is given. Let $\lambda_t=\lambda_t(\left\{ Y_t \right\},\theta)$ (with the same abbreviation) be the stationary solution as displayed in Eq~\eqref{lambda_sol}. $\tilde{\lambda}_t$ can be regarded as a stationary approximation, which is used in practical estimation. Let $\ell_t(\theta)=-\lambda_t(\theta)+Y_t\log(\lambda_t(\theta))$ and $\ell=\ell(\theta)=\sum_{t=1}^n \ell_t(\theta)$ be the corresponding quantities calculated from the stationary solution. 

The first major result is the strong consistency of $\hat{\theta}$ in Eq~\eqref{theta_hat} under the two assumptions about the process.

\begin{theorem}
  Under the assumption (A1),
   $\hat{\theta}$ is strongly consistent, i.e.,
 $
    \hat{\theta}\to \theta_0 \; a.s.
 $
  \label{consistency_para}
\end{theorem}

Since the threshold $r$ is integer-valued, the consistency of $\hat{r}$ implies that $\hat{r} = r$ eventually. Therefore, the efficiency of the other estimates with the threshold being estimated together is asymptotically the same as that when the threshold is known.
We henceforth remove $r$ from the parameter vector $\theta$ and only consider a central limit theorem for the maximum likelihood estimator with known threshold $r$.
Under this setting, $\tilde{\ell}$ is differentiable with respect to $\theta$, and the score function can be calculated using the varying-coefficient representation of $\lambda_t$ as in Eq~\eqref{lambda_vc}.

The score function is
\begin{align*}
  \tilde{S}_n(\theta)&=\frac{\partial \tilde{\ell}(\theta)}{\partial \theta}
  =\sum_{t=1}^n (\frac{Y_t}{\tilde{\lambda}_{t}}-1) \frac{\partial \tilde{\lambda}_{t}}{\partial \theta},
\end{align*}
where
\begin{align}
  \frac{\partial \tilde{\lambda}_{t}}{\partial \theta}
=\left( 
\begin{array}{c}
  \frac{\partial \tilde{\lambda}_{t}}{\partial \theta^{(1)}}\\
  \frac{\partial \tilde{\lambda}_{t}}{\partial \theta^{(2)}}
\end{array}
 \right),\label{lambda_deri1}
\end{align}
and
\begin{align}
  \frac{\partial \tilde{\lambda}_{t}}{\partial \theta^{(i)}}
  = (1,\tilde{\lambda}_{t-1},Y_{t-1})^\intercal 1\left\{ Y_{t-1}\in R_i \right\} +a_{t-1} \frac{\partial \tilde{\lambda}_{t-1}}{\partial \theta^{(i)}}, \textrm{ for } i=1,2.\label{lambda_deri2}
\end{align}

Let 
\begin{align*}
  G=\E
    \left[ \frac{1}{\lambda_t}\left(\frac{\partial \lambda_t}{\partial\theta}\right)\left( \frac{\partial \lambda_t}{\partial\theta}
    \right)^\intercal
 \right],
\end{align*}
 then we state the asymptotic normality of the maximum likelihood estimator in the following theorem.

\begin{theorem}
  Under the assumption (A1) except that the threshold $r$ is known, the maximum likelihood estimator $\hat{\theta}=( (\hat{\theta}^{(1)})^\intercal,(\hat{\theta}^{(2)})^\intercal)^\intercal $ is asymptotically normal,
  \begin{align*}
    \sqrt{n}(\hat{\theta}-\theta_0)\xrightarrow{d} N(0, G^{-1}).
  \end{align*}
  \label{asymptotics_theta}
Furthermore, the matrix $G$ can be estimated consistently by 
\begin{align}
\widehat{G}= \frac{1}{n}\sum_{i=1}^n \frac{1}{\tilde{\lambda}_{t}}\left(\frac{\partial\tilde{\lambda}_{t}}{\partial \theta}\right) \left( \frac{\partial \tilde{\lambda}_{t}}{\partial \theta}\right)^\intercal.\label{g_hat}
\end{align}
\end{theorem}

\begin{remark}
  
Since $r\in \mathbb{N}$, $\tilde{\ell}$ is not differentiable with respect to the threshold variable $r$. In practice, the maximization of the log-likelihood function can be done in the following two steps.

\begin{enumerate}[Step (1):]
  \item For each $r\in[0,r_*]\cap \mathbb{N}$, find $\theta^{(i)}_r$ such that
    \[
    (\hat{\theta}^{(1)}_r,\hat{\theta}^{(2)}_r)=\arg \max_{(\theta^{(1)},\theta^{(2)})\in \mathcal{D}} \tilde{\ell}(r,\theta^{(1)},\theta^{(2)}).
    \]
  \item The threshold is estimated by searching over all candidates
  \[
  \hat{r}=\arg \max_{r\in[0,r_*]\cap\mathbb{N}} \tilde{\ell}(r,\hat{\theta}^{(1)}_r,\hat{\theta}^{(2)}_r),
  \]
  and the final estimate for $\theta^{(i)} $ is $\hat{\theta}^{(i)}_{\hat{r}}\;(i=1,2)$.
\end{enumerate}
\end{remark}
\begin{remark}
Since the threshold is searched over the set of candidates $[0,r_*]$, the upper bound $r_*$ should be large enough so that the set includes the true threshold. However, since the computation time of the estimation procedure increases approximately linearly with respect to the number of candidates, $r_*$ cannot be too large when computation resource is limited. Also, when the bound is too broad, there might not be enough number of observations to ensure consistent estimation. A strategy frequently used in practice is to replace the upper bound $r_*$ as well as the lower bound $0$ by some numbers determined based on the data (cf. \citet{ChengLiYuZhouWangLo2011}). Specifically, fix $\alpha_1<\alpha_2\in(0,1)$ and find the empirical $\alpha_i$-th quantile for $Y_t$, $\hat{q}_i$. Then the interval $[0,r_*]$ is replaced by $[\hat{q}_1,\hat{q}_2]$. The choice of the pair $(\alpha_1,\alpha_2)$ can be $(0.2,0.8)$ or more conservatively $(0.1,0.9)$.
\end{remark}

\section{Simulation study and real data analysis}
\label{sec:sim}
We report the simulation study with two sets of parameters and one real data analysis in this section. 

A two-step estimation procedure is applied as indicated in Section~\ref{estimation}. First we fix $\alpha_1=0.2$ and $\alpha_2=0.8$ and find the empirical $\alpha_i$-quantile of $\left\{ Y_i \right\}_{i=1}^n$, $\hat{q}_i\;(i=1,2)$. Then, for a given threshold candidate, $r\in[\hat{q}_1,\hat{q}_2]\cap \mathbb{N}$, we supply the negative log-likelihood function and its gradient to E04UCF, a NAG Fortran subroutine designed to minimize a smooth function subject to constraints, to obtain the parameter estimate $\hat{\theta}_r$ for the given $r$. The final estimate is obtained by selecting $r$ and the corresponding $\hat{\theta}_r$ which minimizes the negative log-likelihood function. 
\subsection{Simulation study}
Two sets of parameters are considered in our simulation. The true parameter values are listed under Table~\ref{tab:s1} and Table~\ref{tab:s2} respectively. The first parameter set has both regimes stationary, while the second one has an explosive lower regime and negative serial dependence, as illustrated in Table~\ref{tab:nacf}.

\begin{table}
 \begin{center}
  \begin{tabular}{cccccc}
    Lag & $1$ 		& $2$ 	& $3$ 	& $4$ 	& $5$ \\
    ACF & $-0.104$ 	& $-0.074$ 	& $0.015$ 	& $-0.047$ 	& $0.099$ 
  \end{tabular}
  \caption{The autocorrelation function of a sample path simulated with the second parameter set and 500 observations.}
  \label{tab:nacf}
\end{center}
\end{table}

We are interested in checking the following points. The estimated threshold is expected to be identical to the true value when sample size is sufficiently large. The parameters for each regime are consistent and asymptotically normal, so we would like to see whether its sample mean and sample variance are close to the true ones. However, since no explicit form for the asymptotic variance is available, its inverse is estimated by $\widehat{G}$ as in Eq~\eqref{g_hat}. For each set of parameters, 1000 sample paths are simulated. Then for each sample path, one estimate of $\theta$, $\hat{\theta}$, and one copy of the asymptotic covariance matrix $\widehat{G}^{-1}$ are obtained. By the asymptotic result and the law of large numbers we have $n\overline{\textrm{cov}}(\hat{\theta}) \approx \overline{\widehat{G}^{ -1 }}$, where $\overline{\widehat{G}^{-1}}$ is the sample mean of $\widehat{G}^{-1}$ over the 1000 replications. The sample covariance matrix is of course dependent on the length of sample path, however, $n\overline{\textrm{cov}}(\hat{\theta})$ and $\overline{\widehat{G}^{-1}}$ should be approximately equal to a constant matrix independent of $n$ provided that $n$ is sufficiently large.

The simulation results for the two sets of parameters are reported in Table~\ref{tab:s1} and Table~\ref{tab:s2} respectively. Some interesting observations can be made. In general, $\hat{r}$ converges to $r$ very fast. However the speed of this convergence seems to depend on other parameters. For the first set of parameters, even when $n$ is as large as $3000$, $\hat{r}$ does not equal to $r$ in rare samples. However, $\hat{r}$ is identical to the true value when sample size is $500$ for the second set of parameters, which is a moderate sample size for a threshold model. 

The consistency and asymptotic variance of the other parameters are confirmed in both examples. The average estimated parameters are close to the true values and the accuracy increases as the sample size increases. However, the intercept parameters $d_{i}$ seem to have large variances, comparing to the other parameters. This phenomenon is also found in the Poisson autoregression model \citep{FokianosRahbekTjostheim2009}.
In the first example, $n\overline{\textrm{cov}}(\hat{\theta})$ and $\overline{\hat{G}^{-1}}$ match each other reasonably well. Such phenomenon is not so apparent in the second example, especially for $d_{i}$. This might be due to the fact that the lower regime is explosive in the second example.

\begin{table}
  \begin{center}
\begin{tabular}{ccccccccc}
Sample size & Description & $r$ & $d_1$ & $a_{1}$ & $b_{1}$ & $d_2$ & $a_2$ & $b_{2}$\\
\hline
&  $\theta_0$ & 7 & 0$\cdot$50 & 0$\cdot$70 & 0$\cdot$20 & 0$\cdot$30 & 0$\cdot$40 & 0$\cdot$50 \\
\hline
\multirow{3}{*}{$n=500$}	
& $\overline{\hat{\theta}}$ & 6$\cdot$80 & 0$\cdot$63 & 0$\cdot$69 & 0$\cdot$18 & 0$\cdot$83 & 0$\cdot$37 & 0$\cdot$47\\
& $n\overline{\textrm{cov}}(\hat{\theta})$& 1100 & 53 & 2$\cdot$34 & 1$\cdot$76 & 416 & 7$\cdot$69 & 6$\cdot$45\\
& $\overline{\widehat{G}^{-1}}$ & N/A    & 40$\cdot$8 & 2$\cdot$03 & 2$\cdot$32 & 444 & 6$\cdot$45 & 5$\cdot$60 \\
\hline
\multirow{3}{*}{$n=1000$}	
& $\overline{\hat{\theta}}$ & 7$\cdot$00 & 0$\cdot$56 & 0$\cdot$70 & 0$\cdot$19 & 0$\cdot$60 & 0$\cdot$38 & 0$\cdot$48\\
& $n\overline{\textrm{cov}}(\hat{\theta})$& 503$\cdot$5 & 34$\cdot$5 & 1$\cdot$85 & 2$\cdot$21 & 433 & 6$\cdot$84 & 6$\cdot$01\\
& $\overline{\widehat{G}^{-1}}$ & N/A     & 28$\cdot$9 & 1$\cdot$73 & 1$\cdot$79 & 405 & 5$\cdot$16 & 5$\cdot$46 \\
\hline
\multirow{3}{*}{$n=2000$}
& $\overline{\hat{\theta}}$ & 7$\cdot$02 & 0$\cdot$53 & 0$\cdot$70 & 0$\cdot$20 & 0$\cdot$42 & 0$\cdot$39 & 0$\cdot$49\\
& $n\overline{\textrm{cov}}(\hat{\theta})$&  123  & 26$\cdot$2 & 1$\cdot$72 & 1$\cdot$90 & 288 & 4$\cdot$80 & 4$\cdot$76\\
& $\overline{\widehat{G}^{-1}}$ & N/A     & 25$\cdot$6 & 1$\cdot$62 & 1$\cdot$63 & 349 & 4$\cdot$78 & 5$\cdot$18 \\
\hline
\multirow{3}{*}{$n=3000$} & $\overline{\hat{\theta}}$ & 7$\cdot$00 & 0$\cdot$52 & 0$\cdot$70 & 0$\cdot$20 & 0$\cdot$37 & 0$\cdot$40 & 0$\cdot$50\\
& $n\overline{\textrm{cov}}(\hat{\theta})$& 5 & 26$\cdot$8 &  1$\cdot$76  &   1$\cdot$76 &  266 &  5$\cdot$33& 4$\cdot$99\\
& $\overline{\widehat{G}^{-1}}$ & N/A  & 24$\cdot$5&  1$\cdot$61&  1$\cdot$61&  332&  4$\cdot$64&  5$\cdot$05 \\
\hline
  \end{tabular}
  \caption{Simulation 1. The true parameters are in the row with description $\theta_0$. For each sample size, 1000 replications are simulated. Then the mean of estimates, sample size times the variance of estimates and mean of asymptotic variances (if available) are reported respectively.}
  \label{tab:s1}
\end{center}
\end{table}

\begin{table}
  \begin{center}
\begin{tabular}{ccccccccc}
Sample size & Description & $r$ & $d_1$ & $a_{1}$ & $b_{1}$ & $d_2$ & $a_2$ & $b_{2}$\\
\hline
&  $\theta_0$ & 6 & 0$\cdot$50 & 0$\cdot$80 & 0$\cdot$70 & 0$\cdot$20 & 0$\cdot$20 & 0$\cdot$10 \\
\hline
\multirow{3}{*}{$n=500$}	
& $\overline{\hat{\theta}}$ & 6$\cdot$00 & 0$\cdot$47 & 0$\cdot$82 & 0$\cdot$69 & 0$\cdot$32 & 0$\cdot$19 & 0$\cdot$09\\
& $n\overline{\textrm{cov}}(\hat{\theta})$ & 0 & 28$\cdot$17 & 3$\cdot$36 & 2$\cdot$56 & 64$\cdot$96 & 1$\cdot$57 & 1$\cdot$74\\
& $\overline{\widehat{G}^{-1}}$ & N/A & 34$\cdot$05 & 3$\cdot$52 & 2$\cdot$52 & 133$\cdot$27 & 1$\cdot$73 & 1$\cdot$48 \\
\hline
\multirow{3}{*}{$n=1000$}	
& $\overline{\hat{\theta}}$ & 6$\cdot$00 & 0$\cdot$50 & 0$\cdot$81 & 0$\cdot$70 & 0$\cdot$28 & 0$\cdot$20 & 0$\cdot$09\\
& $n\overline{\textrm{cov}}(\hat{\theta})$& 0  & 30$\cdot$29 & 3$\cdot$35 & 2$\cdot$40 & 75$\cdot$55 & 1$\cdot$61 & 1$\cdot$27\\
& $\overline{\widehat{G}^{-1}}$ & N/A     & 33$\cdot$65 & 3$\cdot$48 & 2$\cdot$52 & 133$\cdot$54 & 1$\cdot$73 & 1$\cdot$47 \\
\hline
\multirow{3}{*}{$n=2000$}	
& $\overline{\hat{\theta}}$ & 6$\cdot$00 & 0$\cdot$50 & 0$\cdot$80 & 0$\cdot$70 & 0$\cdot$23 & 0$\cdot$20 & 0$\cdot$10\\
& $n\overline{\textrm{cov}}(\hat{\theta})$&  0   & 29$\cdot$36 & 3$\cdot$28 & 2$\cdot$47 & 82$\cdot$68 & 1$\cdot$46 & 1$\cdot$21\\
& $\overline{\widehat{G}^{-1}}$ & N/A    & 33$\cdot$32 & 3$\cdot$45 & 2$\cdot$50 & 133$\cdot$90 & 1$\cdot$74 & 1$\cdot$47 \\
\hline
\multirow{3}{*}{$n=3000$} & $\overline{\hat{\theta}}$ & 6$\cdot$00 & 0$\cdot$50 & 0$\cdot$80 & 0$\cdot$70 & 0$\cdot$22 & 0$\cdot$20 & 0$\cdot$10\\
& $n\overline{\textrm{cov}}(\hat{\theta})$& 0 & 32$\cdot$56 & 3$\cdot$64 & 2$\cdot$53 & 98$\cdot$93 & 1$\cdot$57 & 1$\cdot$43\\
& $\overline{\widehat{G}^{-1}}$  & N/A  & 33$\cdot$12 & 3$\cdot$44 & 2$\cdot$50 & 133$\cdot$65 & 1$\cdot$73 &  1$\cdot$48 \\
\hline
  \end{tabular}
  \caption{Simulation 2. The true parameters are in the row with description $\theta_0$. For each sample size, 1000 replications are simulated. Then the mean of estimates, sample size times the variance of estimates and mean of asymptotic variances (if available) are reported respectively.}
  \label{tab:s2}
\end{center}
\end{table}

\subsection{Analysis of annual counts of major earthquakes in the world}

In this example we study the series of annual counts of major earthquakes with magnitude 7 (inclusive) or above during 1900 -- 2010, which is plotted in Figure~\ref{fig:all_earthquakes}. The data from 1900 to 2006 can be found in page 4 of \citet{ZucchiniMacDonald2009}, and the rest is extracted from 
the website of U.S. Geological Survey. The sample mean and sample variance are 19$\cdot$30 and 50$\cdot$37 respectively, showing considerable over-dispersion. The marginal distribution of $\left\{ Y_t \right\} $ in a self-excited threshold Poisson autoregression is highly expected to be non-Poissonian. It also displays strong positive serial dependence, as can be seen in Figure~\ref{fig:acf_earthquake}.

\begin{figure}
\begin{center}
\includegraphics[scale=.5]{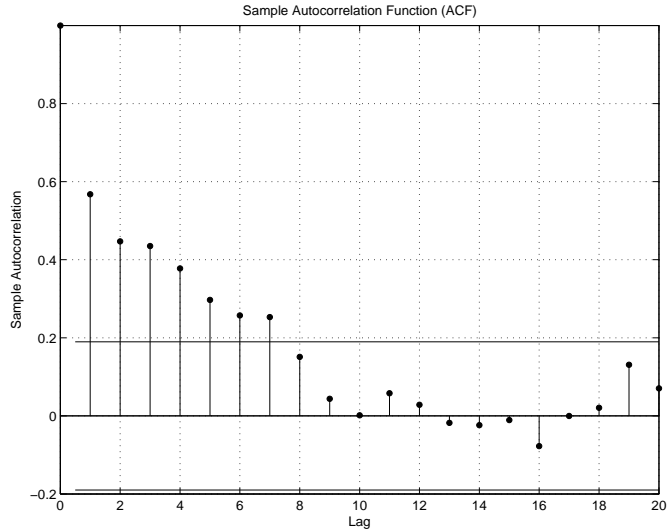}
\end{center}
\caption{ACF of the earthquake data.}
\label{fig:acf_earthquake}
\end{figure}

The series has been studied with hidden Markov models with discrete states by \citet{ZucchiniMacDonald2009}. Here we would like to compare the performances of the Poisson autoregression (PAR) versus the self-excited threshold Poisson autoregression for this data set. It is interesting to note that the (threshold) Poisson autoregression is also a hidden Markov chain but with continuous states. In order to compare the out-of-sample performances, the first 100 observations are used to estimate the parameters, while the last 11 are used to calculate the out-of-sample mean square error (MSE), serving as an assessment to model performance. The estimation results are shown in Table~\ref{tab:summary_earthquakes}. 

The self-excited threshold Poisson autoregression outperforms the ordinary Poisson autoregression according to AIC, in-sample MSE, and out-of-sample MSE. By BIC the Poisson autoregression seems to be better, which is understandable, since BIC is very conservative when selecting models with more parameters. In the threshold case, all parameter estimates are significantly different from zero, except that $d_2 $ is marginally significant and $b_{2}$=0$\cdot$001, which in fact is the lower bound for $b_{2}$ in our algorithm for estimating the parameters. The same threshold model with $b_{2}=0$ is also fitted, but the result remains almost the same, as can be seen in Table~\ref{tab:summary_earthquakes}. The basic statistics of the Pearson's residual which is defined as $(Y_t-\hat{\lambda}_t)/\sqrt{\hat{\lambda}_t}$ under the self-excited threshold Poisson autoregression model are summarized in Table~\ref{tab:pr}, and its ACF is plotted in Figure~\ref{fig:acf_residual}, which shows that there is no virtually significant serial dependence in the residual sequence.

The original data and the fitted series by the two models are plotted in Figure~\ref{fig:all_earthquakes}. It is observed that the threshold model fits the data better when $Y_t$ is large, i.e., its improvement are mainly in the upper regime. If more data were available, a Poisson autoregression with two or more thresholds might be considered. However, insufficiency of data is very likely to result in unreliable parameter estimates, so we content ourselves with the present model. 

A closer look at the fitted parameters reveals the possible different dynamics of the underlying process according to the threshold. Note that the estimated threshold is 25, which is quite large. The difference between the intercepts, $d_1$=3$\cdot$27 versus $d_2$=14$\cdot$33, implies that large number of major earthquakes in one year is very likely to be followed by a lot of earthquakes during the following year. Another notable feature is that $b_{2}=0$, showing that once a large number is observed, the conditional mean of the process would be stably large with less fluctuations comparing to the lower regime in which the conditional mean depends on both the latent mean process and the realized observations. For the earthquake data, this means that more earthquakes will be expected in the next few years once a large number of major earthquakes are observed in a year, as during the years 1942 -- 1950 and 1968 -- 1970.

\begin{table}
\begin{tabular}{ccccc}
& PAR & SETPAR & SETPAR (with $b_{2}=0$) \\
$d_1$ & 2$\cdot$96 (1$\cdot$21) 	& 3$\cdot$27 (1$\cdot$36) & 3$\cdot$27 (1$\cdot$36) \\
$a_{1}$ & 0$\cdot$47 (0$\cdot$11) 	& 0$\cdot$49 (0$\cdot$12) & 0$\cdot$49 (0$\cdot$12) \\
$b_{1}$ & 0$\cdot$39 (0$\cdot$07) 	& 0$\cdot$33 (0$\cdot$10) & 0$\cdot$33 (0$\cdot$10) \\
$d_2$& 		& 14$\cdot$30 (7$\cdot$45) & 14$\cdot$33 (7$\cdot$45)\\
$a_2$& 		& 0$\cdot$52 (0$\cdot$20) & 0$\cdot$52 (0$\cdot$20) \\
$b_{2}$& 			& 0$\cdot$001 (0$\cdot$26) & \\
$r$& & 25 & 25 \\
Average log-likelihood  & 39$\cdot$85 & 39$\cdot$89 & 39$\cdot$89\\
AIC 			& -7883$\cdot$5 & -7885$\cdot$1 & -7887$\cdot$1  \\
BIC 			& -7875$\cdot$7 & -7866$\cdot$9 & -7871$\cdot$5 \\
In-sample MSE & 33$\cdot$12 & 30$\cdot$7 & 30$\cdot$7 \\
Out-of-sample MSE & 13$\cdot$4 & 12$\cdot$8 & 12$\cdot$8\\
\end{tabular}
\caption{Summary of model estimates. Standard errors (if available) are in parenthesis. }
\label{tab:summary_earthquakes}
\end{table}

\begin{figure}
\begin{center}
\includegraphics[scale=0.3]{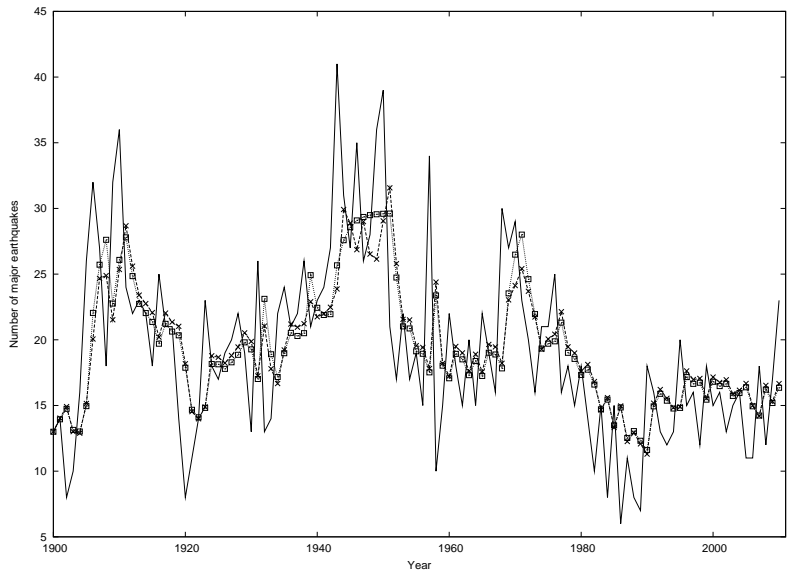}
\end{center}
\caption{Plot of fitted curves of the earthquake data: The original observations are solid, the series fitted by Poisson autoregression is marked by crosses and that fitted by the self-excited threshold Poisson autoregression is marked by squares.}
\label{fig:all_earthquakes}
\end{figure}

\begin{table}
  \begin{center}
\begin{tabular}{cccc}
Mean & Standard error & Skewness & Excess kurtosis\\
-0$\cdot$02 & 1$\cdot$219 & 0$\cdot$537 & 0$\cdot$429
\end{tabular}
\caption{Statistics summary of the Pearson residuals of the earthquake data fitted by the self-excited threshold Poisson autoregression model.}
\label{tab:pr}
\end{center}
\end{table}

\begin{figure}
\begin{center}
\includegraphics[scale=0.35,angle=-90]{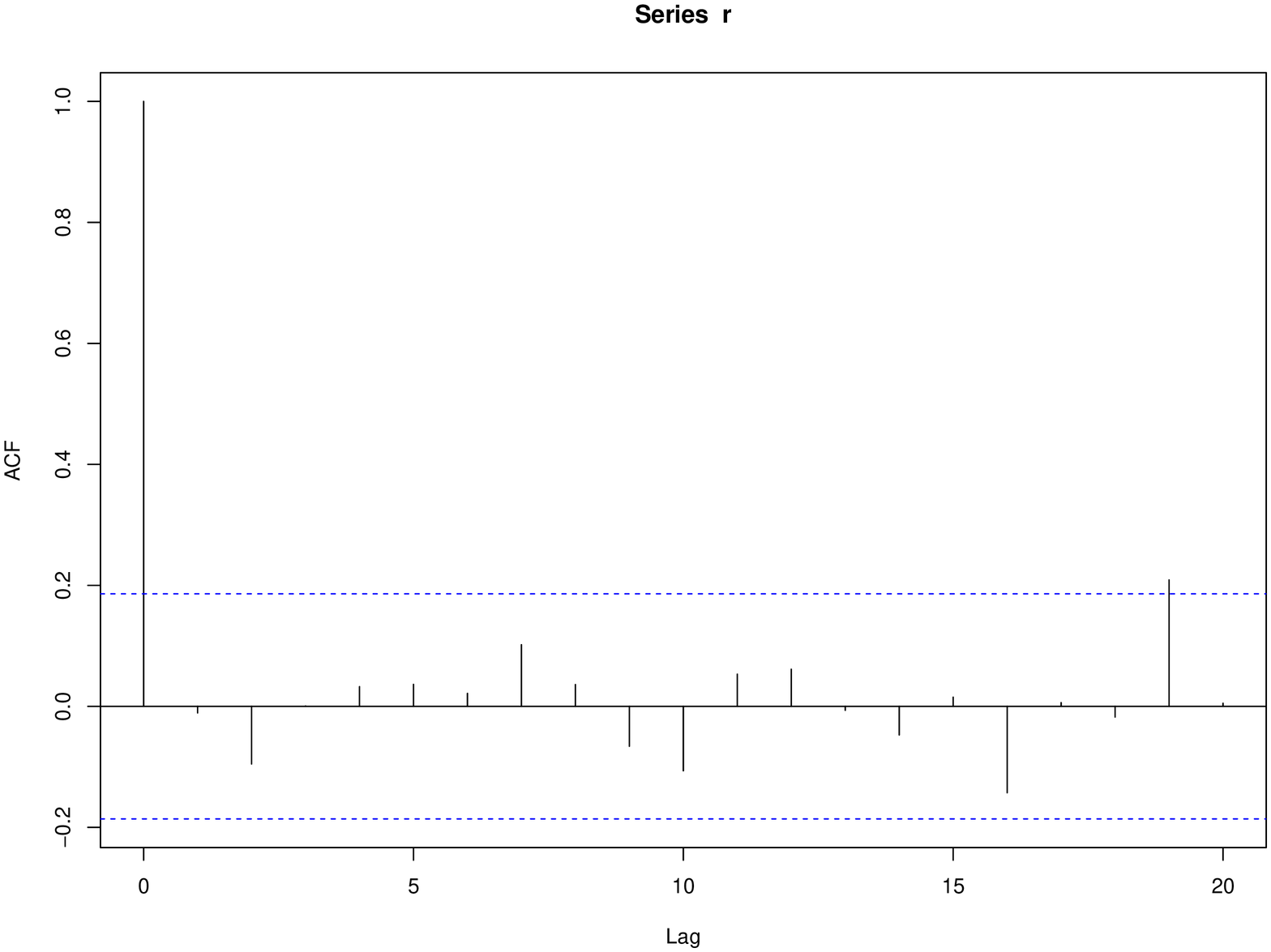}
\end{center}
\caption{ACF of the Pearson residuals of the earthquake data fitted by the self-excited threshold Poisson autoregression model.}
\label{fig:acf_residual}
\end{figure}

\section{Discussion}
\label{sec:concl}

There are some open problems deserving further investigation. The asymptotic properties of the maximum likelihood estimator derived in Theorem~\ref{asymptotics_theta} might be extended to the case without the constraint that $b_1<1$. Another question is to test the self-excited threshold Poisson autoreregression model against the original Poisson autoregression model. Lastly, beyond the self-excited threshold Poisson model discussed in this paper, the following extension with multiple thresholds can be considered. For given integers $0=r_0<r_1\cdots<r_{n-1}<r_n=\infty$, it is assumed that $\mathcal{L}(Y_t \mid \mathcal{F}_{t-1})  = \textrm{Poisson}(\lambda_t)$ , where
\begin{align*}
  \lambda_t=
  \sum_{i=1}^n (d_{i}+a_{i}\lambda_{t-1}+b_{i}Y_{t-1})1\{Y_{t-1}\in [r_{i-1},r_{i})\},
\end{align*}
and $d_{i}> 0, a_{i}> 0,b_{i}> 0\;(i=1,\dots,n)$.

Results similar to Theorem~\ref{thm:stable},Theorem~\ref{consistency_para}, and Theorem~\ref{asymptotics_theta} can be established in a similar manner.

\section{Appendix}

In the following proofs, without explicit specification, $C$ denotes a generic positive constant, and $\rho$ a generic constant such that $\rho\in(0,1)$. $\|X\|_p$ denotes the $L_p$-norm of a random variable $X$. The transition probability kernel of $\left\{ \lambda_t \right\}$ is denoted by $\mathbf{P}$. For any function $V:\mathbb{R}\to\mathbb{R}$, let $\mathbf{P}V(\lambda)=\E(V(\lambda_1)|\lambda_0=\lambda)$.

\subsection{Proof of Theorem~\ref{thm:stable}}

\begin{proof}
We first prove some lemmas.

\begin{lemma}
  For a Poisson process $\{N(u), ~u\ge0\}$ with unit rate, 
  \begin{enumerate}
  \item $\displaystyle \lim_{u\to\infty} N(u)/u =1$ almost surely.
  \item The family of random variables $\{ (\frac{N(u)}{u})^s,~u\ge 1  \}$ 
    is uniformly integrable for any integer $s\geq 1$.
  \end{enumerate}
\label{lem1}
\end{lemma}

\begin{proof}
The first assertion is clearly correct for integer-valued $u$'s
following the law of large numbers. For arbitrary $u$, let  $\lfloor u \rfloor$ be the integer part of $u$, then
$  \lfloor u\rfloor \le u <  \lfloor u\rfloor+1$, 
and 
$ N(\lfloor u\rfloor)\le N(u)\le N(\lfloor u\rfloor+1)$. 
The conclusion follows.

For the second assertion, since $N(u)$ has a Poisson distribution with mean $u$, its $q$-th order moment is a polynomial function of $u$ of degree $q$. Therefore there exists a constant $C$ such that
\[  \E\left( \frac{N(u)}{u} \right)^q \le C, \; u\ge 1.
\]
For given order $s\ge 1$, using the bound with $q>s$ the uniformly integrability of the family $\left\{ [N(u)/u]^s, \; u \ge 1 \right\}$ is proved.
\end{proof}

\begin{lemma}\label{lem:Lyap}
For $s\ge 1$,  let $V(\la)=\la^s$.  Then 
\[       \lim_{\la\to\infty} \frac{\mathbf{P} V(\la)}{V(\la)} = (a_2+b_2)^s.
  \]
\end{lemma}

\begin{proof}
  
We have 
\begin{eqnarray*}
\frac{\mathbf{P}V(\lambda)}{V(\lambda)} &=& \frac{\E\left[V(\lambda_1)\mid \lambda_0=\lambda\right]}{V(\lambda)}\\
&=& \E\left[  
\left( \frac{d_1}{\lambda} + a_1+b_1\frac{Y_0}{\lambda} \right)^s1_{\{Y_0\le r\}} 
+ \left( \frac{d_2}{\lambda} + a_2+b_2\frac{Y_0}{\lambda} \right)^s 1_{\{Y_0>r\}} 
\right] \\
& := & \E[ h(\la,\omega)]~.
\end{eqnarray*}

For fixed $\omega$ and when $\la\to\infty$,  since by Lemma~\ref{lem1},
$Y_0/\lambda=N_0( \lambda)/\lambda \to  1$
a.s.,  $ 1_{\{Y_0\le r \}}  \to 0 $.  Therefore $h(\lambda,\omega)\to (a_2+b_2)^s$ a.s. as $\lambda\to\infty$.

Next we check the  uniform integrability condition.  Using $(a+b)^s\le 2^{s-1}(a^s+b^s)$
for $s\ge1,~a,b\ge0$, it is clear that for all $\lambda\in[1,\infty)$,
\[     0\le h(\lambda,\omega)\le c(s) \left( 1+ \left(\frac{Y_0}{\lambda}\right)^s   \right),
\]
for some constant $c(s)$ independent of $\lambda$ (but 
depending on $s$ and the parameters).
By Lemma~\ref{lem1}, 
the family $ \{(Y_0/\lambda)^s,~\lambda\ge 1\}  $ is uniformly integrable, so is
the family $\{h(\lambda,\omega),~\lambda\ge 1 \}$.   We thus obtain the announced limit.
\end{proof}

\begin{lemma}\label{lem:Feller}
    The Markov chain $\left\{ \lambda_t \right\}$ is weakly Feller.
\end{lemma}

\begin{proof}
  To make the dependence on Poisson processes explicit, we write the state equation Eq~\eqref{lambda} in the form $\lambda_t=F(\lambda_{t-1},N_{t-1})$ with $Y_{t-1}$ replaced by $N_{t-1}\left( \lambda_{t-1} \right)$, using the representation of $Y_{t-1} = N_{t-1}(\lambda_{t-1})$ in Eq~\eqref{y}.
Let $g:\R_+\to\R$ be any continuous and bounded function. We need to prove that 
$\mathbf{P}g(x)=\E[g(\lambda_1) \mid \lambda_0=x]$ is continuous. Let $\varepsilon>0$ and first
choose $\eta>0$ such that $2\|g\|_\infty(1-e^{-2\eta})\le \varepsilon/2$.
Consider a  neighbourhood $(x_0-\eta, x_0+\eta]$
of some  $x_0\in\R_+$.  Define the event 
\[    A = \left\{  
\textrm{\normalfont ~the Poisson process
$N_0$ has no jumps in  $(x_0-\eta,x_0+\eta]$ } 
\right\}. 
\]
Clearly, $P(A)=e^{-2\eta}$.  Write 
\begin{align*}
\mathbf{P}g(x)-\mathbf{P}g(x_0)
=& \E\left[ g(F(x,N_0))- g(F(x_0,N_0)) \right] \\  
=&  \E\left[ \left\{g(F(x,N_0)) -g(F(x_0,N_0))  \right\}1_A  \right] \\
& + \E\left[ \left\{ g(F(x,N_0)) -g(F(x_0,N_0))  \right\}  1_{A^c}        \right].
\end{align*}
On $A^c$, we have 
\[ 
\left|\E\left\{ g(F(x,N_0)) - g(F(x_0,N_0)) \right\}1_{A^c} \right| \le 
2\|g\|_\infty P(A^c) = 2\|g\|_\infty (1-e^{-2\eta}) \le \varepsilon/2.
\]
And on the event $A$, $N_0(x)=N_0(x_0)$, for any $x\in(x_0-\eta,x_0+\eta]$. The mapping $x\mapsto F(x,N_0)$ is continuous, so is $ x\mapsto g(F(x,N_0))1_A$ which is also  bounded.
Thus by Lebesgue's dominated convergence theorem, 
\[    \E\left[\left\{  g(F(x,N_0)) -  g(F(x_0,N_0)) \right\} 1_A \right] \to  0,~x\to x_0.
\]
We can then choose $\eta_1<\eta$ such that for $|x-x_0|<\eta_1$, 
\[ 
\left| \E\left\{  g(F(x,N_0)) -  g(F(x_0,N_0)) \right\} 1_A 
\right| \le  \varepsilon/2~.
\]
Finally  for $|x-x_0|<\eta_1$, by collecting these two estimates,  
\[   | \mathbf{P}g(x)-\mathbf{P}g(x_0)|\le\varepsilon.
\]
The proof is complete. 
\end{proof}

 \begin{lemma}\label{lem:e-chain}
    The Markov chain $\left\{ \lambda_t \right\}$ is an e-chain provided that $a_1<1$ and $a_2+b_2<1$.
\end{lemma}
\begin{proof}
  
It suffices to show that for any continuous function $f$ with compact support and $\epsilon>0$, there exists an $\eta>0$ such that $|\mathbf{P}^kf(x) - \mathbf{P}^kf(z)|<\epsilon$, for any $|x-z|<\eta$ and all $k\ge 1$, where $\mathbf{P}^kf(\cdot)=\E(f(\lambda_k)\mid \lambda_0=\cdot)$. 

Without loss of generality, assume $|f|\le 1$. Take $\epsilon'$ and $\eta$ sufficiently small such that $\epsilon'+4\eta / (1-\bar{a})<\epsilon$, where $\bar{a}=\max\{a_1,a_2\}<1$, and $|f(x_1) - f(z_1)| < \epsilon'$ whenever $|x_1 - z_1|<\eta$. Denote $p(\cdot\mid x)$ as the probability mass function of a Poisson distribution with intensity $x$. Then for the case when $k=1$,
\begin{eqnarray*}
  && |\mathbf{P}f(x_1)-\mathbf{P}f(z_1)|\\
  &\le& |\displaystyle\sum_{i=0}^r f(d_1 + a_1 x_1+b_1 i)p(i\mid x_1) - \sum_{i=0}^r f(d_1 + a_1 z_1 +b_1 i)p(i\mid z_1)|\\
                          &&+ |\sum_{j=r+1}^{\infty}f(d_2 +a_2 x_1 + b_2 j)p(j\mid x_1) - \sum_{j=r+1}^{\infty} f(d_2 +a_2 z_1 + b_2 j)p(j\mid z_1)|\\
                         &:=& \Rmnum{1} + \Rmnum{2}.
\end{eqnarray*}
For $x_1\ge z_1$, 
\begin{eqnarray*}
\displaystyle\sum_{i=0}^{\infty}|p(i\mid x_1)-p(i\mid z_1)|&=&\sumi |\frac{x_1^i e^{-x_1}}{i!}-\frac{z_1^i e^{-z_1}}{i!}|\\
                           &\le&\sumi \frac{(x_1^i- z_1^i) e^{-x_1}}{i!} + \sumi \frac{z_1^i (e^{-z_1} - e^{-x_1})}{i!}\\
                           &=& 2(1-e^{-|x_1-z_1|}).
\end{eqnarray*}

The same inequality holds for $x_1<z_1$ by symmetry.
Hence for any $x_1$ and $z_1$, we have
\begin{eqnarray}
\label{eq:DiffPoisPdf}
\sumi |p(i\mid x_1) - p(i\mid z_1)| \le 2(1 - e^{-|x_1 - z_1|}).
\end{eqnarray}
It follows that
\begin{eqnarray*}
\Rmnum{1}&\le&\displaystyle\sum_{i=0}^r |f(d_1+a_1 x_1 + b_1 i) - f(d_1 + a_1 z_1 + b_1 i)|p(i \mid x_1)\\
                &&+\sum_{i=0}^r |f(d_1 + a_1 z_1 + b_1 i)||p(i \mid x_1) - p(i \mid z_1)|\\
               &\le& \epsilon' F(r\mid x_1) + 2(1 - e^{-|x_1-z_1|}),
\end{eqnarray*}
where $F(r \mid x_1)=\sum_{i=0}^r p(i \mid x_1)$. The last inequality follows from Eq~\eqref{eq:DiffPoisPdf}, $|f|\le 1$, and the fact that $|(d_1+a_1 x_1 + b_1 i) - (d_1+a_1 z_1 + b_1 i)|=a_1|x_1-z_1|<\eta$. It follows from a similar argument that $\Rmnum{2}\le \epsilon' (1-F(r \mid x_1)) + 2(1-e^{-|x_1-z_1|})$. Hence we have 
\begin{eqnarray}
|\mathbf{P}f(x_1) - \mathbf{P}f(z_1)|\le \epsilon' + 4(1 - e^{-|x_1-z_1|}),
\label{eq:echaink1}
\end{eqnarray}
for $|x_1-z_1|<\eta$. For the case when $k=2$, it follows from 
\begin{eqnarray*}
\E\{f(\lambda_2) \mid \lambda_0=x\}=\E\{\E[f(\lambda_2) \mid \lambda_1]\bigr|\lambda_0=x\}
\end{eqnarray*}
that
\begin{eqnarray*}
  |\mathbf{P}^2 f(x_1) - \mathbf{P}^2 f(z_1)| &=& |\mathbf{P}(\mathbf{P}f)(x_1) - \mathbf{P}(\mathbf{P}f)(z_1)|\\
  &\le& |\displaystyle\sum_{i=0}^r p(i \mid x_1)\mathbf{P}f(x_2^{(1)}) - \sum_{i=0}^r p(i \mid z_1) \mathbf{P}f(z_2^{(1)})|\\
  &&+ |\displaystyle\sum_{j=r+1}^{\infty} p(j \mid x_1)\mathbf{P}f(x_2^{(2)}) - \sum_{j=r+1}^{\infty} p(j \mid z_1) \mathbf{P}f(z_2^{(2)})|\\
                                       &:=& \Rmnum{3} + \Rmnum{4},
\end{eqnarray*}
where $x_2^{(1)}=d_1 + a_1 x_1 + b_1 i, x_2^{(2)}=d_2 + a_2 x_1 + b_2 j, z_2^{(1)}=d_1 + a_1 z_1 + b_1 i$, and $z_2^{(2)}=d_2 + a_2 z_1 + b_2 j$. Then 
\begin{eqnarray*}
  \Rmnum{3}&\le&\displaystyle\sum_{i=0}^r p(i \mid x_1)|\mathbf{P}f(x_2^{(1)}) - \mathbf{P}f(z_2^{(1)})| + \sum_{i=0}^r |\mathbf{P}f(z_2^{(1)})||p(i \mid x_1) - p(i \mid z_1)|\\
  &\le&\left\{\epsilon' + 4\left(1 - e^{-|x_2^{(1)} - z_2^{(1)}|}\right)\right\} F(r \mid x_1) + 2\left(1 - e^{-|x_1 - z_1|}\right),
\end{eqnarray*}
which follows from (\ref{eq:DiffPoisPdf}) and (\ref{eq:echaink1}). Similarly, we have 
$$\Rmnum{4}\le \left\{\epsilon' + 4(1 - e^{-|x_2^{(2)} - z_2^{(2)}|})\right\}(1 - F(r \mid x_1)) + 2\left(1 - e^{-|x_1 - z_1|}\right).$$

Since $|x_2^{(1)} - z_2^{(1)}|=a_1 |x_1-z_1|$ and $|x_2^{(2)} - z_2^{(2)}|=a_2 |x_1-z_1|$, so by letting $\bar{a}=\max\{a_1, a_2\}$, we have
\begin{eqnarray*}
  |\mathbf{P}^2f(x_1) - \mathbf{P}^2f(z_1)|&\le& \epsilon' + 4\left(1 - e^{-\bar{a}|x_1 - z_1|}\right) + 4\left(1 - e^{-|x_1 - z_1|}\right).
\end{eqnarray*}
Inductively, one can show that for any $k\ge 1$,
\begin{eqnarray*}
  |\mathbf{P}^kf(x_1) - \mathbf{P}^kf(z_1)| &\le& \epsilon' + 4 \sum_{s=0}^{k-1} \left(1 - e^{-\bar{a}^s |x_1 - z_1|}\right)\\
                                        &\le& \epsilon' + 4\sum_{s=0}^{\infty}\bar{a}^s|x_1 - z_1|\\
                                       &\le& \epsilon' + \frac{4\eta}{1-\bar{a}}<\epsilon,
\end{eqnarray*}
where the second inequality holds since $1-e^{-x}\le x$. Hence $\{\lambda_t\}$ is an e-chain.

\end{proof}

\noindent{\bf Proof of Theorem~\ref{thm:stable}\quad}
By Lemma~\ref{lem:Lyap}, for any initial value  $\la_0=x$, 
the sequence of transition probabilities 
\[ \overline\pi_n(x,dy) = \frac1n\left\{ \mathbf{P}(x,dy)+\cdots+\mathbf{P}^n(x,dy)  \right\}
\]
is tight \citep[Proposition 2.1.6]{Duflo1997}. Moreover, using the weak Feller property established in Lemma~\ref{lem:Feller},
we know that the weak limit of any subsequence of $\{
\overline\pi_n(x,dy)\}$ is 
an invariant probability measure of $\mathbf{P}$. 

Then note that $\lambda^{\ast}=d_1 / (1-a_1)$ is a reachable state by letting $Y_1=Y_2=\ldots=Y_t=0$ for large $t$. Combined with the fact that $\left\{ \lambda_t \right\}$ is an e-chain, it follows that the stationary distribution is unique.

The fact that $\mu(|x|^s)<\infty$ for all $s\ge0$ directly results from  the Lyapounov property established in Lemma~\ref{lem:Lyap}. The strong law of large numbers also follows from this method, see Proposition 6.2.12 and the remarks in Section 6.2.2 in \citet{Duflo1997}. The proof is complete.
\end{proof}

\subsection{Proof of Corollary~\ref{prop:stable_y}}
\begin{proof}
  The stability of the joint process is clear. To see $Y_t\in L_s$, for all $s>0$, it suffices to note that $\lambda_t\in L_s$ for all $s>0$ and the following fact
 \[
 \E (Y_t)^s=\E[ \E\{ (Y_t)^s \mid \lambda_t\}]=(\E(Poly(\lambda_t,s))<\infty,
 \]
 where $Poly(\lambda_t,s)$ is the polynomial of $\lambda_t$ of order $s$ which represents the $s$th moment of a Poisson random variable with mean $\lambda_t$.

\end{proof}

\subsection{Proof
of Theorem~\ref{consistency_para}}
\begin{proof}
  Since the log-likelihood $\tilde{\ell}$ is calculated with a given initial value $\tilde{\lambda}_1$, we first show that
  the log-likelihood $\tilde{\ell}$ is asymptotically independent of $\tilde{\lambda}_1$.

  Using the varying-coefficient representation in Eq~\eqref{lambda_vc}, we have
  \begin{align*}
    \lambda_t(\lambda_1)&=\sum_{k=1}^{t-2} \prod_{j=1}^{k-1}a_{t-j} c_{t-k} +\prod _{j=1}^{t-1}a_{t-j}\lambda_1,
  \end{align*}
  which implies
  \begin{align*}
    \sup_{\theta\in\mathcal{D}}|\lambda_t(\lambda_1)-\tilde{\lambda}_t(\tilde{\lambda_1})|
    =\sup_{\theta\in\mathcal{D}}|\prod_{j=1}^{t-1}a_{t-j}(\lambda_1-\tilde{\lambda}_1)|
    \leq  K \rho^t,
  \end{align*}
  where $\rho=\sup_{\theta\in\mathcal{D}}\max\left\{ a_{1},a_{2} \right\}<1$ and $K=|\lambda_1-\tilde{\lambda}_1|/\rho$.

  Then the difference between the log-likelihoods based on arbitrary initial value and on the stationary initial one is
  \begin{align*}
    \sup_{\theta\in\mathcal{D}}|\frac{1}{n}(\ell(\lambda_1)-\ell(\tilde{\lambda}_1)|
    =&\sup_{\theta\in\mathcal{D}}|\frac{1}{n}\sum_{t=1}^n Y_t(\log(\lambda_t)-\log(\tilde{\lambda}_t))-(\lambda_t-\tilde{\lambda}_t)|\\
    =&\sup_{\theta\in\mathcal{D}}|\frac{1}{n}\sum_{t=1}^n Y_t\log(1+\frac{\lambda_t-\tilde{\lambda}_t}{\tilde{\lambda}_t})-(\lambda_t-\tilde{\lambda}_t)|\\
    \leq&\sup_{\theta\in\mathcal{D}}\frac{1}{n}\sum_{t=1}^n Y_t|\frac{\lambda_t-\tilde{\lambda}_t}{\tilde{\lambda}_t}|+|\lambda_t-\tilde{\lambda}_t|\\
    \leq&\sup_{\theta\in\mathcal{D}}\frac{1}{n}\sum_{t=1}^n|\lambda_t-\tilde{\lambda}_t|(\frac{Y_t}{d_0}+1)\\
    \leq&\frac{1}{n}\sum_{t=1}^nK\rho^t(\frac{Y_t}{d_0}+1)\\
    \to& 0, \; a.s.
  \end{align*}
where $d_0=\inf_{\theta\in\mathcal{D}}\min\{d_1,d_2\}>0$.

The $a.s.$ limit holds because of the Ces\`aro lemma and the observation that $\rho^tY_t\to 0, a.s.$ (see also \citet{FrancqZakoian2004}).

Secondly, we prove that $\E [\ell_t(\theta)]$ is continuous in $\theta$. Since $r$ is discrete, we need only to prove the following property. For any $\theta\in\mathcal{D}$, let $V_\eta(\theta)=B(\theta,\eta)$ be an open ball centered at $\theta$ with radius $\eta$, then
\begin{align}
	\E\left( \sup_{\tilde{\theta}\in V_\eta(\theta)}|\ell_t(\tilde{\theta})-\ell_t(\theta)| \right)\to0, \textrm{ as }\eta\to0.
	\label{continuity_theta}
\end{align}

 To see this, observe that 
 $$|\ell_t(\tilde\theta)-\ell_t(\theta)|\leq(\frac{Y_t}{\lambda_t(\tilde{\theta})}+1)|\lambda_t(\tilde{\theta})-\lambda_t(\theta)|,$$
 and
\begin{align*}
	|\lambda_t(\theta)-\lambda_t(\tilde{\theta})|
	=&|\sum_{k} \prod_{j=1}^{k-1}a_{t-j} c_{t-k}-\prod_{j=1}^{k-1}\tilde{a}_{t-j} \tilde{c}_{t-k}|\\
	=&|\sum_{k} (\prod_{j=1}^{k-1}a_{t-j}-\prod_{j=1}^{k-1}\tilde{a}_{t-j}) c_{t-k}+\prod_{j=1}^{k-1}\tilde{a}_{t-j}(c_{t-k}- \tilde{c}_{t-k})|\\
	\leq & C \eta \sum_{k} \rho^{k}(1+Y_{t-k}).
\end{align*}

Then 
\begin{align*}
\E\left( \sup_{\tilde{\theta}\in V_\eta(\theta)}|\ell_t(\tilde{\theta})-\ell_t(\theta)| \right)
\leq & \|\frac{Y_t}{d_0}+1\|_2 \|\lambda_t-\tilde{\lambda}_t\|_2\\
\leq & C \eta \|\frac{Y_t}{d_0}+1\|_2 \sum_k \rho^k\|Y_t\|_2\\
\to & 0, \textrm{ as }\eta \to 0.
\end{align*}

Next, we check the model identifiability. By Jensen inequality, we have
\begin{align*}
  \E\left[\ell_t(\theta)-\ell_t(\theta_0)\right]
  =&\E\left[\E\left(\log\frac{\phi(Y_t \mid \lambda_t(\theta))}{\phi(Y_t \mid \lambda_t(\theta_0))} \mid \mathcal{F}_{t-1}\right)\right]\\
  \leq &\E\left[\log \E\left(\frac{\phi(Y_t \mid \lambda_t(\theta))}{\phi(Y_t \mid \lambda_t(\theta_0))} \mid \mathcal{F}_{t-1}\right)\right]\\
  = &\E(\log(1))=0,
\end{align*}
where $\phi(\cdot \mid y) $ denotes the Poisson distribution function with mean $y$, and the equality holds iff $\lambda_t(\theta)=\lambda_t(\theta_0)\; a.s.~\mathcal{F}_{t-1}$.

Suppose that $\tilde{\theta}$ satisfies $\tilde{\lambda}_t=\lambda_t(\tilde{\theta})=\lambda_t(\theta_0)\; a.s.~\mathcal{F}_{t-1}$. Without loss of generality, assume $\tilde{r}\geq r$. For ease of notation, let $\lambda_t=\lambda_t(\theta_0)$ temporarily, then conditional on $\mathcal{F}_{t-2}$, we have $\tilde{\lambda}_{t-1}=\lambda_{t-1}\; a.s.$, and almost surely
\begin{align}
  \tilde{\lambda}_t-\lambda_t=
  &(\tilde{d}_{t-1}+\tilde{b}_{t-1}Y_{t-1}+\tilde{a}_{t-1}\tilde{\lambda}_{t-1})-(d_{t-1}+b_{t-1}Y_{t-1}+a_{t-1}\lambda_{t-1})\nonumber\\
  =&[(\tilde{d}_1-d_1)+(\tilde{b}_{1}-b_{1})Y_{t-1}+(\tilde{a}_1-a_1)\lambda_{t-1}]1\left\{ Y_{t-1}\leq r \right\}\nonumber\\
  &+[(\tilde{d}_1-d_2)+(\tilde{b}_{1}-b_{2})Y_{t-1}+(\tilde{a}_{1}-a_2)\lambda_{t-1}]1\left\{ r< Y_{t-1}\leq \tilde{r} \right\}\nonumber\\
  &+[(\tilde{d}_{2}-d_2)+(\tilde{b}_{2}-b_{2})Y_{t-1}+(\tilde{a}_{2}-a_2)\lambda_{t-1}]1\left\{ \tilde{r}< Y_{t-1}\right\}.\label{diff_lambda}
\end{align}

Note that $\mathcal{F}_{t-1}=\sigma\left\{ Y_{t-1},\mathcal{F}_{t-2} \right\}$, $Y_{t} \mid \lambda_t \sim \textrm{Poisson}(\lambda_t)$,  it can be seen from Eq~\eqref{diff_lambda} that if $\tilde{\lambda}_t-\lambda_t=0\; a.s.\; \mathcal{F}_{t-1}$, we must have $\tilde{\theta}=\theta_0$.

Now we are ready to prove the consistency. Consider an arbitrary (small) open neighbourhood of $\theta_0$, say $V$,  then for any $\vartheta\in V^c\cap \mathcal{D}$, we have $\E [\ell_t(\vartheta)]< \E [\ell_t(\theta_0)]$, since $V^{c}\cap \mathcal{D}$ is compact and $\E [\ell_t(\theta)]$ is continuous in $\theta$, we have $\kappa=\E [\ell_t(\theta_0)]-\sup_{\theta\in V^{c}\cap \mathcal{D}}\E [\ell_t(\theta)] >0 $. And for any $\theta\in V^{c}\cap \mathcal{D}$, there exists $\eta_{\theta}>0$ such that $\E[\sup_{\vartheta\in V_{\eta_\theta}(\theta)}\ell_t(\theta)] < \E [\ell_t(\theta)]+\frac{1}{6}\kappa$. Also by the compactness of $V^{c}\cap \mathcal{D}$, there exists a finite open cover of $V^c\cap \mathcal{D}$, say, $\{V_{\eta_{\theta_j}}(\theta_j),j=1,\dots,m\}$. For any $\theta\in \mathcal{D}$ and $k\gg 0$,
\begin{align*}
	&\varlimsup_{n\to\infty} \sup_{\theta^*\in V_{1/k}(\theta)\cap \mathbf{\Theta}} \frac{1}{n}\tilde{\ell}(\theta^*)\\
	\leq &\varlimsup_{n\to\infty} \sup_{\theta^*\in V_{1/k}(\theta)\cap \mathbf{\Theta}} \frac{1}{n}\ell(\theta^*)+\varlimsup_{n\to\infty} \sup_{\theta^*\in V_{1/k}(\theta)\cap \mathbf{\Theta}} \frac{1}{n}|\ell(\theta^*)-\tilde{\ell}(\theta^*)|\\
	\leq &\varlimsup_{n\to\infty} \frac{1}{n} \sum_{t=1}^n\sup_{\theta^*\in V_{1/k}(\theta)\cap \mathbf{\Theta}} \ell_t(\theta^*).
\end{align*}

By Corollary~\ref{prop:stable_y} and as in \citet{FrancqZakoian2004}, we have almost surely for $n\gg 0$ and $j=1,\dots,m$,
\begin{align*}
	\sup_{\theta\in V\eta_{\theta_j}(\theta_j)} \frac{1}{n} \sum_{t=1}^{n} \tilde{\ell}_t(\theta)
	&\leq \sup_{\theta\in V\eta_{\theta_j}(\theta_j)} \frac{1}{n} \sum_{t=1}^{n} \ell_t(\theta)+\frac{1}{6}\kappa\\ 
	&\leq  \frac{1}{n} \sum_{t=1}^{n} \sup_{\theta\in V\eta_{\theta_j}(\theta_j)}\ell_t(\theta)+\frac{1}{6}\kappa \\
&\leq \E \left(\sup_{\theta\in V\eta_{\theta_j}(\theta_j)}\ell_t(\theta)\right)+\frac{1}{3}\kappa \\
&\leq \E [l_{t}(\theta_0)]-\frac{2}{3}\kappa.
\end{align*}
And
\begin{align*}
	\sup_{\theta\in V} \frac{1}{n} \sum_{t=1}^{n} \tilde{\ell}_t(\theta)\geq \frac{1}{n} \sum_{t=1}^{n} \tilde{\ell}_t(\theta_0)\geq \frac{1}{n} \sum_{t=1}^{n} \ell_t(\theta_0) -\frac{1}{6}\kappa \geq \E [\ell_t(\theta_0)]-\frac{1}{3}\kappa.
\end{align*}
Therefore, for any (small) neighbourhood of $\theta_0$, $V$, for $n\gg 0$, we have almost surely
\begin{align*}
	\sup_{\theta\in V\eta_{\theta_j}(\theta_j)} \frac{1}{n} \sum_{t=1}^{n} \tilde{\ell}_t(\theta)\leq \sup_{\theta\in V} \frac{1}{n} \sum_{t=1}^{n} \tilde{\ell}_t(\theta),
\end{align*}
which implies $\hat{\theta}\in V$.

%
%
\end{proof}

\subsection{Proof of Theorem~\ref{asymptotics_theta}}
We here only give an outline of the proof, a detailed proof can be found in the supplementary material.
\begin{proof}
  By Taylor's expansion, for $j=1,\dots,6$, there exists some $\theta_{ (j) }$  between $\theta_0$ and $\hat{\theta}$ such that
  \begin{align*}
    0=\frac{1}{\sqrt{n}}\sum_{t=1}^n \frac{\partial{\tilde{\ell}_t(\hat{\theta})}}{\partial \theta_j}
    =\frac{1}{\sqrt{n}} \sum_{t=1}^n \frac{\partial \tilde{\ell}_t(\theta_0)}{\partial\theta_j}
    +\left(\frac{1}{n} \sum_{t=1}^n \frac{\partial^2 \tilde{\ell}_t(\theta_{(j)})}{\partial \theta_j \partial\theta^\intercal}\right) \sqrt{n}(\hat{\theta}-\theta_0).
  \end{align*}
The theorem follows if it can be proved that
  \begin{align*}
    \frac{1}{\sqrt{n}}\sum_{t=1}^n\frac{\partial\tilde{\ell}_t(\theta_0)}{\partial\theta}\xrightarrow{d} N(0,G),
  \end{align*}
  and
  \begin{align*}
    \frac{1}{n} \sum_{t=1}^n \frac{\partial^2 \tilde{\ell}_t(\theta^*)}{\partial \theta \partial\theta^\intercal}\xrightarrow{p} -G,
  \end{align*}
for all $\theta^* $ between $\theta_0$ and $\hat{\theta}$.

To show these, we prove the following statements,
  \begin{enumerate}
\renewcommand{\theenumi}{(S\arabic{enumi})}
    \item \label{s1} $\frac{1}{\sqrt{n}}\sum_{t=1}^n\frac{\partial \ell_t(\theta_0)}{\partial \theta}\xrightarrow{d} N(0,G)$.
    \item $\|\frac{1}{\sqrt{n}} \sum_{t=1}^n (\frac{\partial \ell_t(\theta_0)}{\partial \theta}-\frac{\partial \tilde{\ell}_t(\theta_0)}{\partial \theta})\|\xrightarrow{p} 0$.\label{s2}
    \item There exists a neighbourhood of $\theta_0$, $V(\theta_0)$, such that for all $i,j,k\in\left\{ 1,\dots,6 \right\}$,
   \begin{align*}
     \E \left( \sup_{\theta\in V(\theta_0)} |\frac{\partial^3 \ell_t(\theta)}{\partial \theta_i \partial \theta_j \partial \theta_k}| \right)<\infty.
   \end{align*}
\label{s3}
    \item For the neighbourhood  $V(\theta_0)$ specified above,
    \begin{align*}
      \sup_{\theta\in V(\theta_0)} \|\frac{1}{n}\sum_{t=1}^n \left( \frac{\partial^2 \ell_t(\theta)}{\partial \theta \partial \theta^\intercal }-\frac{\partial^2 \tilde{\ell}_t(\theta)}{\partial \theta \partial \theta^\intercal }\right) \|\xrightarrow{p} 0.
    \end{align*}
\label{s4}
    \item  $\frac{1}{n}\sum_{t=1}^n\frac{\partial^2 \ell_t(\theta^*)}{\partial \theta \partial \theta^\intercal}\xrightarrow{a.s.} -G$, uniformly for all $\theta^* $ between $\theta_0$ and $\hat{\theta}$.\label{s5}
  \end{enumerate}

\end{proof}

\bibliographystyle{jtsa1}
\bibliography{ref}

\newpage
\section*{Supplementary material}
\subsection*{ Complementary for establishing the statements in the proof of Theorem~\ref{asymptotics_theta} }
 We write $\lambda_t$ as in Eq~\eqref{lambda_vc}, then  
\begin{align*}
    \frac{\partial \ell_t}{\partial \theta}=(\frac{Y_t}{\lambda_t}-1)\frac{\partial \lambda_t}{\partial \theta},
  \end{align*}
  and
  \begin{align}
    \frac{\partial \lambda_t}{\partial \theta}=&\left(
    \begin{array}{c}
    \frac{\partial \lambda_t}{\partial \theta^{(1)}}\\
    \frac{\partial \lambda_t}{\partial \theta^{(2)}}\\
  \end{array}
  \right),\label{lambda_deri}
 \end{align}
  with
  \begin{align*}
    \frac{\partial \lambda_t}{\partial \theta^{(i)}}=&
    \left(\begin{array}{l}1\\ \lambda_{t-1}\\ Y_{t-1}\end{array}\right)1\left\{ Y_{t-1}\in R_i \right\}
      +a_{t-1} \frac{\partial \lambda_{t-1}}{\partial \theta^{(i)}}
      \quad( i=1,2).
  \end{align*}


The derivative in Eq~\eqref{lambda_deri} can be written in a compact form as
  \begin{align*}
    \frac{\partial \lambda_t}{\partial \theta}
    :=&\nu_{t-1}+a_{t-1}\frac{\partial \lambda_{t-1}}{\partial \theta}
    =\sum_{k\geq 1}(\prod_{j=1}^{k-1} a_{t-j}) \nu_{t-k}.
  \end{align*}
By assumption $a_t\leq\max\left\{ a_{1},a_2 \right\}=a_M<1$, then
  \begin{align*}
    \frac{\partial \lambda_t}{\partial \theta}\leq \sum_k a_M^{k-1} \nu_{t-k}.
  \end{align*}
  In particular, we have
  \begin{align}
    \frac{\partial \lambda_t}{\partial d_{i}}
    =\sum_{k\geq 1}(\prod_{j=1}^{k-1} a_{t-j})1\left\{ Y_{t-1}\in R_{i} \right\} \leq \sum_{k\geq 1} a_M^{k-1}
    \leq \frac{1}{1-a_M}.
    \label{alpha0bound}
  \end{align}
  Writing $\lambda_t=\sum_{k\geq 1} (\prod_{j=1}^{k-1}a_{t-j}) c_{t-k}$ with $c_{t}=d_t+b_tY_t$, we have
\begin{align*}
    \frac{\partial \lambda_t}{\partial b_{i}}
    =\sum_{k\geq 1}(\prod_{j=1}^{k-1} a_{t-j})\frac{\partial b_{t-k}}{\partial b_{i}}Y_{t-k}
    =\sum_{k\geq 1}(\prod_{j=1}^{k-1} a_{t-j})1\left\{ Y_{t-k}\in R_{i} \right\}Y_{t-k},
\end{align*}
which implies
\begin{align}
    \|\frac{\partial \lambda_t}{\partial b_{i}}\|_2
    \leq& \|Y_{t}\|_2\sum_{k\geq 1} a_M^k.
    \label{betabound}
  \end{align}
Also,
\begin{align}
    \frac{\partial \lambda_t}{\partial a_{i}}
    =\sum_{k\geq 1}\frac{\partial (\prod_{j=1}^{k-1} a_{t-j})}{\partial a_{i}}c_{t-k}    \leq \sum_{k\geq 1}\frac{k-1}{a_{i}} (\prod_{j=1}^{k-1} a_{t-j})c_{t-k}\nonumber,
\end{align}
implies
\begin{align}
    \E\left( \frac{\partial \lambda_t}{\partial a_{i}} \right)
    \leq &\sum_{k\geq 1}\frac{k-1}{a_{i}}a_M^{k-1} (d_M+b_M \E (Y_t))< \infty\label{alpha1bound1},
  \end{align}
    where $ d_M=\max\left\{ d_1,d_2 \right\}, b_M=\max\left\{ b_{1},b_{2}\right\}$, and
  \begin{align}
    \|\frac{\partial \lambda_t}{\partial a_{i}}\|_2
    \leq &\sum_{k\geq 1}\frac{k-1}{a_{i}}a_M^{k-1} (d_M+b_M \|Y_t\|_2)< \infty.
\label{alpha1bound2}
  \end{align}

  Note that
  \begin{align*}
\E\left[\frac{\partial \ell_t(\theta_0)}{\partial \theta}\right]
=\E\left[ \left(\frac{Y_t}{\lambda_t}-1\right)\frac{\partial \lambda_t}{\partial \theta}\right]
=\E\left[\E\left(\frac{Y_t}{\lambda_t}-1\right)\frac{\partial \lambda_t}{\partial \theta}|\mathcal{F}_{t-1}\right]
=0.
  \end{align*}
  Since $\lambda_t$ is bounded from zero, $\lambda_t\geq d_0=\min\left\{ d_1,d_2 \right\}$, with the results in Eq~\eqref{alpha0bound}, Eq~\eqref{betabound}, Eq~\eqref{alpha1bound1}, and Eq~\eqref{alpha1bound2} we have
\begin{align*}
  \textrm{var}\left[\frac{\partial \ell_t(\theta_0)}{\partial \theta}\right]
  =&\E\left[ \left(\frac{Y_t}{\lambda_t}-1\right)^2\left(\frac{\partial \lambda_t}{\partial \theta}\right)\left(\frac{\partial \lambda_t}{\partial \theta}\right)^\intercal\right]\\
  =&\E \left[ \E\left\{ \left(\frac{Y_t}{\lambda_t}-1\right)^2\left(\frac{\partial \lambda_t}{\partial \theta}\right)\left(\frac{\partial \lambda_t}{\partial \theta}\right)^\intercal \mid \mathcal{F}_{t-1}\right\}\right]\\
  =&\E\left[ \frac{1}{\lambda_t}\left(\frac{\partial \lambda_t}{\partial \theta}\right)\left(\frac{\partial \lambda_t}{\partial \theta}\right)^\intercal\right]\\
  =&G<\infty.
\end{align*}

It can be seen that $G$ is non-degenerate (cf. \citet{FrancqZakoian2004}).

Since $\left\{ \partial \ell_t(\theta_0)/\partial \theta \right\} $ is a $L_4$ martingale difference, by the Cram\'er-Wold device and the central limit theorem in Theorem 18.1 of \citet{Billingsley2011} we have the weak convergence,
\begin{align*}
  \frac{1}{\sqrt{n}}\sum_{t=1}^n\frac{\partial \ell_t(\theta_0)}{\partial\theta}\xrightarrow{d} N(0,G).
  \end{align*}

  Then we shall prove Statement \ref{s2}.
To show this, note that for $i=1,2$,
  \begin{align}
    \frac{\partial \tilde{\lambda}_t}{\partial d_{i}}
    =&\sum_{k\geq 1}^{t-2}(\prod_{j=1}^{k-1} a_{t-j})1\left\{ Y_{t-k}\in R_{i} \right\}+\prod_{j=1}^{k-1}a_{t-j} \frac{\partial \tilde{\lambda}_{1}}{\partial d_{i}},\\
    \frac{\partial \tilde{\lambda}_t}{\partial a_{i}}
    =&\sum_{k=1}^{t-2}\frac{\partial (\prod_{j=1}^{k-1} a_{t-j})}{\partial a_{i}}c_{t-k}+\prod_{j=1}^{t-1}a_{t-j} \frac{\partial \tilde{\lambda}_1}{\partial a_{i}},\\
    \frac{\partial \tilde{\lambda}_t}{\partial b_{i}}
    =&\sum_{k=1}^{t-2} (\prod_{j=1}^{k-1} a_{t-j})Y_{t-k}1\left\{ Y_{t-k}\in R_{i} \right\} +\prod_{j=1}^{t-1}a_{t-j} \frac{\partial \tilde{\lambda}_1}{\partial b_{i}}.
  \end{align}
Since $\partial \tilde{\lambda}_1/\partial \theta$ can be regarded as a fixed value, we have
\begin{align*}
  \sup_{\theta\in \mathcal{D}}\|\frac{\partial \tilde{\lambda}_t}{\partial \theta}
  -\frac{\partial \lambda_t}{\partial \theta}\|
  \leq C \rho^t, a.s.
\end{align*}
Note that we also have $|\lambda_t-\tilde{\lambda}_t|\leq C\rho^t$, which implies $|\frac{1}{\lambda_t}-\frac{1}{\tilde{\lambda}_t}|\leq C\rho^t$, for $\lambda_t $ and $ \tilde{\lambda}_t$ are bounded from 0. Note that
\begin{align*}
  \frac{\partial \ell_t(\theta_0)}{\partial \theta}-\frac{\partial \tilde{\ell}_t(\theta_0)}{\partial \theta}
  = & \left(\frac{Y_t}{\lambda_t(\theta_0)}-1\right)\frac{\partial \lambda_t(\theta_0)}{\partial \theta}- \left(\frac{Y_t}{\tilde{\lambda}_t(\theta_0)}-1\right)\frac{\partial \tilde{\lambda}_t(\theta_0)}{\partial \theta}\\
  =& Y_t\left[ \left(\frac{1}{\lambda_t}-\frac{1}{\tilde{\lambda}_t}\right)\frac{\partial \lambda_t}{\partial \theta}+\frac{1}{\tilde{\lambda}_t}\left(\frac{\partial \lambda_t}{\partial \theta}-\frac{\partial \tilde{\lambda}_t}{\partial \theta}\right)\right]-\left(\frac{\partial \lambda_t}{\partial \theta}-\frac{\partial \tilde{\lambda}_t}{\partial \theta}\right).
\end{align*}
Then it is readily seen that
\begin{align*}
  \|\frac{\partial \ell_t(\theta_0)}{\partial \theta}-\frac{\partial \tilde{\ell}_t(\theta_0)}{\partial \theta}\|
  \le & C\rho^t \left[1+Y_t \left(1+ \|\frac{\partial \lambda_t}{\partial \theta}\|\right)\right].
\end{align*}
Note that $\E (Y_t\|\partial \lambda_t(\theta_0)/\partial \theta\|)<\infty$, then for any $\varepsilon>0$,
\begin{align*}
  \pr \left(\|\frac{1}{\sqrt{n}}\sum_{t=1}^n\left(\frac{\partial \ell_t(\theta_0)}{\partial \theta}-\frac{\partial \tilde{\ell}_t(\theta_0)}{\partial \theta}\right)\|>\varepsilon\right)
  \leq& \frac{1}{\sqrt{n} \varepsilon}\sum_{t=1}^n C\rho^t\left[ 1+\E (Y_t)+ \E \left(\|Y_t \frac{\partial \lambda_t}{\partial \theta}\|\right)\right]\\
 \to & 0, \textrm{ as } n\to\infty.
\end{align*}

Next we will prove Statement \ref{s3}. Through direct calculation, we obtain
\begin{align}
  \frac{\partial^3 \ell_t(\theta)}{\partial \theta_i \partial \theta_j \partial \theta_k}
  &=\left(-\frac{Y_t}{\lambda_t^2}\right)\left(
  \frac{\partial^2\lambda_t}{\partial\theta_i\partial\theta_j}\frac{\partial\lambda_t}{\partial\theta_k}+
  \frac{\partial^2\lambda_t}{\partial\theta_i\partial\theta_k}\frac{\partial\lambda_t}{\partial\theta_j}+
  \frac{\partial^2\lambda_t}{\partial\theta_j\partial\theta_k}\frac{\partial\lambda_t}{\partial\theta_i}\right)\nonumber\\
  &+2\frac{Y_t}{\lambda_t^3}\frac{\partial\lambda_t}{\partial\theta_i}\frac{\partial\lambda_t}{\partial\theta_j}\frac{\partial\lambda_t}{\partial\theta_k}
  +\left(\frac{Y_t}{\lambda_t}-1\right)\frac{\partial^3\lambda_t}{\partial\theta_i\partial\theta_j\partial\theta_k}.
\label{3rd_derivative}
\end{align}

Consider, for example, $\partial^3 \ell_t(\theta)/\partial a_{1}^3$. Write $\lambda_t=\sum_k \prod_{j=1}^{k-1}a_{t-j} c_{t-k}$, then for $i=1,2,3$,
\begin{align*}
\frac{\partial^i \lambda_t(\theta)}{\partial a_{1}^i}
=\sum_{k\geq 1} \frac{\partial^i (\prod_{j=1}^{k-1}a_{t-j})}{\partial a_{1}^i} c_{t-k}
\leq \sum_{k\geq 1} \frac{(k-1)\cdots(k-i)}{a_{1}^i}(\prod_{j=1}^{k-1}a_{t-j})c_{t-k}.
\end{align*}
We may select $ V(\theta_0)$ small enough such that $a_M=\sup_{\theta\in V(\theta_0)}\max\{a_{1},a_2\}<1$, and $a_m=\inf_{\theta\in V(\theta_0)}\min\{a_{1},a_2\}>0$, then 
\begin{align*}
\frac{\partial^i \lambda_t(\theta)}{\partial a_{1}^i}
\leq &\sum_{k\geq 1} \frac{(k-1)\cdots(k-i)}{a_m^i}a_M^{k-1}c_{t-k}\quad (i=1,2,3).
\end{align*}

Recall that $c_t=d_t+a_tY_t$, then it is easily seen that there exist constants $ \zeta_{t,i}>0 $, such that $\sum_{t}\zeta_{t,i}<\infty$, and
\begin{align*}
  \sup_{\theta\in V(\theta_0)} \frac{\partial^i \lambda_t(\theta)}{\partial a_{1}^i}
  \leq \zeta_{0,i}+\sum_{k\geq 1} \zeta_{k,i} Y_{t-k}
  :=\mu_{t,i}.
\end{align*}


From Eq~\eqref{3rd_derivative}, we have
\begin{align*}
  \E\left( \sup_{\theta\in V(\theta_0)}|\frac{\partial^3 \ell_t(\tilde{\theta})}{\partial \theta_i \partial \theta_j \partial \theta_k}| \right)\leq \E\left[3\frac{Y_t}{d_m^2}\mu_{t,2}\mu_{t,1}+2\frac{Y_t}{d_m^3}\mu_{t,1}^3+\left(\frac{Y_t}{d_m}+1\right)\mu_{t,3}\right].
\end{align*}

The expression on the right-hand-side of the inequality can be proved to be finite, if $\mu_{t,3}\in L_2$, $\mu_{t,1}\in L_6$, $\mu_{t,2}\in L_4$, which can be verified by Minkowski inequality and the fact that $Y_t\in L_p$, for all $p> 0$.

As for the second order derivative in Statement \ref{s4}, note that
similar to the case for the first order derivative, we can prove
\begin{align}
  \sup_{\theta\in\mathbf{\Theta}}\|  \frac{\partial^2 \lambda_t}{\partial\theta \partial\theta^\intercal}-\frac{\partial^2 \tilde{\lambda}_t}{\partial\theta \partial\theta^\intercal}\|\leq C\rho^t.
  \label{2d}
\end{align}

It is easily seen that
\begin{align*}
  \frac{\partial^2 \ell_t}{\partial\theta \partial\theta^\intercal}=&\left(\frac{Y_t}{\lambda_t}-1\right)\frac{\partial^2 \lambda_t}{\partial\theta \partial\theta^\intercal}-\frac{Y_t}{\lambda_t^2}\left(\frac{\partial \lambda_t}{\partial \theta}\right)\left(\frac{\partial \lambda_t}{\partial \theta}\right)^\intercal,
\end{align*}
and
\begin{align*}
  \E\left(\frac{\partial^2 \ell_t(\theta_0)}{\partial\theta \partial\theta^\intercal}\right)=-G.
\end{align*}

Then
\begin{align*}
  &\frac{\partial^2 \ell_t}{\partial\theta_i \partial\theta_k}
  -\frac{\partial^2 \tilde{\ell}_t}{\partial\theta_i \partial\theta_k}\\
  =& Y_t\left[ \left(\frac{1}{\lambda_t}-\frac{1}{\tilde{\lambda}_t}\right)\frac{\partial^2\lambda_t}{\partial\theta_i\partial\theta_k}+\frac{1}{\tilde{\lambda}_t}\left(\frac{\partial^2 \lambda_t}{\partial\theta_i \partial\theta_k}
  -\frac{\partial^2 \tilde{\lambda}_t}{\partial\theta_i \partial\theta_k}\right)+\left(\frac{1}{\lambda_t^2}-\frac{1}{\tilde{\lambda}_t^2}\right)\frac{\partial\lambda_t}{\partial\theta_i}\frac{\partial\lambda_t}{\partial\theta_k}\right.\\
  &\left.+\frac{1}{\tilde{\lambda}_t^2}\left\{\frac{\partial \lambda_t}{\partial \theta_i}\left(\frac{\partial \lambda_t}{\partial \theta_j}-\frac{\partial \tilde{\lambda}_t}{\partial \theta_j}\right)+\frac{\partial \tilde{\lambda}_t}{\partial \theta_j}\left(\frac{\partial \lambda_t}{\partial \theta_i}-\frac{\partial \tilde{\lambda}_t}{\partial \theta_i}\right) \right\}\right]+\left(\frac{\partial^2 \lambda_t}{\partial\theta_i \partial\theta_k}
  -\frac{\partial^2 \tilde{\lambda}_t}{\partial\theta_i \partial\theta_k}\right).
\end{align*}
Thus, we have
  \begin{align*}
  |\frac{\partial^2 \ell_t}{\partial\theta_i \partial\theta_k}
  -\frac{\partial^2 \tilde{\ell}_t}{\partial\theta_i \partial\theta_k}|
  \leq  C\left[1+Y_t\left(\frac{\partial^2\lambda_t}{\partial\theta_i\partial\theta_k}+\frac{\partial\lambda_t}{\partial\theta_i}\frac{\partial\lambda_t}{\partial\theta_k}+\frac{\partial\lambda_t}{\partial\theta_i}+\frac{\partial\lambda_t}{\partial\theta_k}\right)\right]\rho^t.
\end{align*}
Let 
$$\Gamma_t=\frac{\partial^2\lambda_t}{\partial\theta_i\partial\theta_k}+\frac{\partial\lambda_t}{\partial\theta_i}\frac{\partial\lambda_t}{\partial\theta_k}+\frac{\partial\lambda_t}{\partial\theta_i}+\frac{\partial\lambda_t}{\partial\theta_k},$$
then it can be seen that around a neighbourhood of $\theta_0$, without loss of generality, assuming the same $V(\theta_0)$, we have
$  \sup_{\theta\in V(\theta_0)}\E\left(\Gamma_tY_t\right)<\infty.
$

Similar as in the argument for Statement~\ref{s3}, we can obtain the following by Markov inequality,
\begin{align}
  \sup_{\theta\in\mathbf{\Theta}}|\frac{1}{n}\sum_{t=1}^n \left(  \frac{\partial^2 \ell_t}{\partial\theta_i \partial\theta_j}-\frac{\partial^2 \tilde{\ell}_t}{\partial\theta_i \partial\theta_j}\right)|
  \xrightarrow{p}0.
  \label{}
\end{align}

Lastly, we prove Statement \ref{s5}. Recall that $\theta^* $ lies between $\theta_0$ and $\hat{\theta}$. Consider the Taylor expansion of the second-order derivatives of $\ell_t$ at $\theta_0$, we have
\begin{align*}
  \frac{1}{n}\sum_t \frac{\partial^2 \ell_t(\theta^*)}{\partial\theta_i \partial \theta_j}
  =\frac{1}{n}\sum_t \frac{\partial^2 \ell_t(\theta_0)}{\partial\theta_i \partial \theta_j}+\frac{1}{n}\sum_t\frac{\partial^3 \ell_t(\tilde{\theta})}{\partial\theta_i \partial \theta_j \partial \theta}(\theta^*-\theta_0),
\end{align*}
for some $\tilde{\theta}$ between $\theta_0$ and $\theta^*$. Then the almost sure convergence of $\tilde{\theta}$ to $\theta_0$, the ergodic theorem in Corollary~\ref{prop:stable_y}, and 
Statement~\ref{s3} imply that
\begin{align*}
  \varlimsup \sup_{\theta\in V(\theta_0)}\|\frac{1}{n} \sum_t \frac{1}{n} \frac{\partial^3 \ell_t(\theta)}{\partial\theta_i \partial \theta_j \partial \theta}\| < \infty, a.s.
\end{align*}
Then we have
\begin{align*}
  \lim_{n\to\infty}\frac{1}{n}\sum_t \frac{\partial^2 \ell_t(\theta^*)}{\partial\theta_i \partial \theta_j}
  =\lim_{n\to\infty}\frac{1}{n}\sum_t \frac{\partial^2 \ell_t(\theta_0)}{\partial\theta_i \partial \theta_j}= -G(i,j)\; a.s.
\end{align*}

The proof is complete.

\end{document}